\documentclass{article}

\usepackage{arxiv}

\usepackage[utf8]{inputenc} 
\usepackage[T1]{fontenc}    
\usepackage[hidelinks]{hyperref}       
\usepackage{url}            
\usepackage{booktabs}       
\usepackage{amsfonts}       
\usepackage{nicefrac}       
\usepackage{microtype}      
\usepackage{graphicx}
\usepackage{doi}

\usepackage{amsmath}
\usepackage{amssymb}
\usepackage[mathcal]{euscript}

\usepackage{bm}
\usepackage{epsfig}
\usepackage{color}
\usepackage{cite}
\usepackage{amsthm}

\usepackage{enumitem}

\usepackage{subcaption}
\usepackage{cleveref}

\DeclareCaptionLabelSeparator{enspace}{\enspace}
\newtheorem{remark}{Remark}
\newtheorem{proposition}{Proposition}

\newtheorem{theorem}{Theorem}

\makeatletter \@addtoreset{theorem}{section} \makeatother

\makeatletter \@addtoreset{proposition}{section} \makeatother

\makeatletter \@addtoreset{remark}{section} \makeatother

\title{Conditions for Oscillator Small-Signal Amplitude-Phase Orthogonality.}

\author{Torsten Djurhuus \thanks{The authors are
with the Institute of Physics, Goethe University of Frankfurt am
Main, Max-von- Laue-Strasse 1, 60438, Frankfurt am Main.
(correspondence e-mail:
t.djurhuus@physik.uni-frankfurt.de).} \\
    Goethe-University Frankfurt\\
    \texttt{t.djurhuus@physik.uni-frankfurt.de} \\
    \And
    Viktor Krozer \\
    Goethe-University Frankfurt\\
    \texttt{krozer@physik.uni-frankfurt.de}
}

\renewcommand{\headeright}{\textit{arXiv} Preprint}
\renewcommand{\undertitle}{\textit{arXiv} Preprint}

\hypersetup{ pdftitle={Conditions for Oscillator Small-Signal
Amplitude-Phase Orthogonality.}, pdfsubject={eess.SY, math.CA,
math.DS}, pdfauthor={Torsten Djurhuus, Viktor Krozer},
pdfkeywords={oscillators, phase noise, AM-PM noise conversion,
circuit analysis, nonlinear dynamical systems, system analysis and
design}, }

\begin{document}

\maketitle

\begin{abstract} The paper explores a previously unknown connection relating
the symmetry properties of an oscillator steady-state to the
orthogonal representation of amplitude and phase variables in the
small-signal regime. It is shown that only circuits producing
perfectly symmetric steady-states can produce an orthogonal
Floquet decomposition. Considering room temperature operation this
scenario implies zero AM-PM noise conversion. This surprising and
novel result follows directly from the predictions of a rigorous
model framework first described herein. The work presented in this
text extend the current state-of-the-art w.r.t. oscillator
small-signal/noise characterization.
\end{abstract}


\section{Introduction}

\label{sec0}

The oscillator small-signal/linear-response (LR) governs the
circuit dynamics in reply to small perturbations, \emph{e.g.}
noise, around the periodic steady-state (PSS). Rigorous methods
for characterizing oscillator dynamics in a noisy environment are
absolutely critical for developing analysis/synthesis tools used
to optimize performance of various critical circuits found in
modern communication systems.
\par
In the general case, small-signal oscillator amplitude and phase
variables are defined in-terms of mutually oblique (\emph{i.e.}
non-orthogonal) Floquet vectors decomposing the LR map
\cite{kartner1990,demir2000,coram2001,djurhuus2009}. One important
implication of this oblique representation is the integration of
amplitude-noise into oscillator phase response; a process known as
oscillator AM-PM noise conversion\footnote{note that this implies
a representation where the noise perturbing the oscillator is
decomposed in-terms of an orthogonal frame (see discussion in
\cref{sec1b} for details).
\label{sec0:foot2}}\cite{razavi1996,samori1998,laloue2003,chang1997,bonnin2012,
bonnin2014,traversa2011,djurhuus2005,djurhuus2005_2,djurhuus2006,djurhuus2021}.
The aim of the model described herein is to seek an answer to the
open question :\emph{what type of oscillators support a fully
orthogonal small-signal representation implying zero AM-PM ?} This
scenario is represented mathematically in-terms of an
\emph{orthogonal Floquet decomposition} (OFD) of the LR. The topic
was previously briefly studied by the authors in
\cite{djurhuus2021}, however, only for the special case of planar
oscillator and strictly from a simulation-based perspective.
\par
The paper documents and validates the novel SYM-OFD model
framework with specific focus paid to the remarkable statement :
${\footnotesize \text{ORTHOGONALITY (OFD)} \Rightarrow}$
${\footnotesize\text{STEADY-STATE SYMMETRY}}$. This statement is
both notable and unanticipated. To our knowledge, this constitutes
the first ever description of a direct analytical link relating a
specific decomposition of the oscillator LR and the properties of
the underlying PSS. The statement also provides a formal
sufficient condition for zero AM-PM noise conversion in higher
dimensional oscillators. This relation cannot be reached using
arguments based on empirical or phenomenological reasoning but
relies on the rigorous methodology developed herein. It is
important to note that the above statement references a strict
one-way relation \emph{i.e.} an orthogonal LR representation
\emph{implies} symmetry. The reverse implication is however
\underline{not true} as symmetry does not imply
orthogonality\footnote{this is easily seen by considering simple
2D counter-examples of the form (polar coordinates) $\dot{r} =
F(r), \dot{\phi} = 1 + G(r)$ with $F(a) = 0, dF(a)/dr < 0$ and
$dG(a)/dr \neq 0$. This class of systems generate symmetric
limit-cycle at $r=a$ but \underline{does not} produce an OFD (see
discussion in \cref{sec1}). \label{sec0:foot1}}.
\par
The validity of the orthogonal model representation has been
debated in the literature for several decades (see \emph{e.g.}
\cite{kartner1990,kaertner89,coram2001,djurhuus2009,bonnin2012,bonnin2014}).
Unlike the natural Floquet description, an orthogonal model
representation is, in the general case, artificial(un-natural)
meaning that is coordinate dependent which implies that
corresponding model operators will not represent tensors. This
issue has important implications for oscillator LR modelling. A
coordinate-independent/tensor approach, by definition, always
leads to simplest, cleanest and most generalized model description
\cite{djurhuus2009}. Our work herein provides a definitive
resolution of the aforementioned open debate : in-order for a
orthogonal modelling approach to be valid (\emph{i.e.} be natural,
coordinate-independent) the underlying PSS must be symmetric
(orthogonality implies symmetry).
\par
In answering these types of open question, using a rigorous,
proof-based, methodology, the novel SYM-OFD framework advances the
current state-of-the-art \emph{w.r.t.} time-domain oscillator
small-signal/noise characterization and modelling. It introduces
several novel ideas and insights such as \emph{e.g.} why OFD's are
not observed in real-life oscillator descriptions where non-linear
device-models make it impossible to attain perfect PSS symmetry.
Finally, the methodology enables a whole new category of numerical
optimization tools which could potentially find use in future
applications.
\par
In-order to briefly explore the last point, consider \emph{e.g.} a
scenario where some oscillator \emph{figure-of-merit} (FoM),
describing a particular performance metric of interest, attains an
optimum at an OFD state\footnote{an example here could \emph{e.g.}
be the simple quadrature VCO (QVCO) oscillator circuit which was
shown attain optimum performance when tuned to such an
configuration\cite{djurhuus2005,djurhuus2005_2,djurhuus2006}.}. We
introduce two strictly positive scalars, $\Lambda,\Upsilon$, which
measure the deviation from PSS symmetry, and OFD solution state,
respectively, with $\Upsilon=0,\Lambda=0$ corresponding to perfect
symmetry/OFD (see \cref{sec5a}). In the vicinity of zeros for
these two measures, the theory then predicts that $\Upsilon$
(orthogonality) will be minimized alongside $\Lambda$ (symmetry)
\emph{i.e.} $\Upsilon \to 0 \Rightarrow \Lambda \to 0$. Using
standard minimization routines\cite{galassi2002}, $\min(\Lambda)$
is then derived, starting from a given initial condition in
parameter-space. If parameter sets, achieving a set $\Upsilon$
goal, indeed exist then the SYM-OFD methodology guarantees that
at-least one of these can be found in the vicinity of equivalent
symmetry (zero) points. Due to the one-way nature of the
orthogonality/symmetry relation there will be the possibility of
false-positives (symmetry point may imply a non-OFD state).
However, even with this serious drawback, the scheme proposed here
is easily several orders-of-magnitudes faster than any competing
approach\footnote{the described algorithm is, to our knowledge,
the only one of its kind so the only alternative is brute-force
sampling of parameter points; which scales exponentially with the
dimension of the parameter-space. In comparison, the algorithm
described above scales polynomially (due to minimization
algorithm).\label{sec0:foot3}}. The idea discussed here is only
possible due to the connections forged by the SYM-OFD framework.
It is simply not possible to directly optimize/minimize the OFD
measure, $\Upsilon$, given that the LR equations are unknown
a-priori to calculating the PSS.

\section{Detailed paper summary \& main result.}

\label{sec0x}

\Cref{sec1} starts with a quick introduction to some basic
underlying topics (tangent-bundle, Floquet theory, OFD oscillators
\emph{etc.}). The analysis in the next two sections, leading up to
the main result in \cref{sec3:theo1}, can then be divided into $4$
consecutive steps (refer to acronym/symbol lists in back of paper)
:

\begin{enumerate}
\item \textbf{\Cref{sec2,sec2a}} : A normal-form oscillator
(NF-OSC), $\mathbf{o} = (\psi_{\tau},\xi)$, on the domain
$\mathbb{U}$, is used to \underline{parameterize} the oscillator
under investigation, $\mathbf{q} = (\phi_{\tau},\gamma)$, on
domain $\mathbb{W}_s$ (see \cref{sec2a:fig1,sec2b:fig2}). This
parametrization is facilitated in-terms of a unique
\emph{conjugation-map}, $h : \mathbb{U} \to \mathbb{W}_s$, where
$h$ \emph{conjugates} $\mathbf{o}$ and $\mathbf{q}$; written
$\mathbf{o} \sim_h \mathbf{q}$. Conjugation preserves invariant
sets and specifically $h(\xi) = \gamma$. \item
\textbf{\Cref{sec2b}} : The NF-OSC, $\mathbf{o}$, is chosen
according to the two criteria : it must have a canonical (simple)
model description and it must be a OFD oscillator ($\mathbf{o} \in
\mathbf{O}^{\perp}$). The chosen model is given the handle
PNF-OSC. It is shown that PNF-OSC limit-cycle is perfectly
symmetric $\xi \in \mathcal{SYM}_n$. \item \textbf{\Cref{sec2c}} :
Let $\mathbf{o} \sim_h \mathbf{q}$, where, $\mathbf{o} \in
\mathbf{O}^{\perp}$, is the PNF-OSC developed in \textbf{step \#2}
above. The demand, $\mathbf{q} \in \mathbf{O}^{\perp}$, restricts
the conjugation map as $h \in \mathbf{H}_C \subset \mathbf{H}$
(see \cref{sec2c:prop1} \& symbol list). In summary : let
$\mathbf{q}$ be conjugate to the PNF-OSC, then, $\mathbf{q}\in
\mathbf{O}^{\perp}$, will hold if, and only if, $h \in
\mathbf{H}_C$. A new conjugation operator (equivalence relation),
$\overset{c}{\sim}_h$, is introduced and the above statement is
written (see \cref{sec2c:prop2}) as $\mathbf{q} \in
\mathbf{O}^{\perp} \Leftrightarrow \mathbf{o} \overset{c}{\sim}_h
\mathbf{q}$. \item \textbf{\Cref{sec3:prop2}} : Using Liouville's
and Schottkys theorems
\cite{astala2008,hartman1947,flanders1966,jacobowitz1991,kuhnel2007,iwaniec1998}
it is shown that the conjugation map, $h \in \mathbf{H}_C$,
\underline{must} have the form $h(y) = \rho A y$ where $\rho \in
\mathbb{R} \setminus \{0\}$ is a non-zero real value and $A \in
O(n)$ is an orthogonal $n\times n$ matrix. \item
\textbf{\Cref{sec3:theo1}} : The result is reached by following
the chain of steps \#1-\#4 described above : \textbf{\#1 ($h(\xi)
= \gamma$) $\to$ \#2 ($\xi \in \mathcal{SYM}_n$) $\to$ \#3 ($h \in
\mathbf{H}_C$) $\to$ \#4 ($h(y) = \rho A y$)} which allows the
calculation $\gamma = h(\xi) = A\xi \in \mathcal{SYM}_n$ where
$A\xi \in \mathcal{SYM}_n$ holds because multiplication by the
orthogonal map $A$ preserves the symmetry of $\xi$ (\emph{i.e.}
maps $\mathcal{SYM}_n \to \mathcal{SYM}_n$). At this point we have
reached the main result (\cref{sec3:theo1}) : $\mathbf{q} \in
\mathbf{O}^{\perp} \Rightarrow \gamma \in \mathcal{SYM}_n$ which,
in words, says that only oscillators with a perfectly symmetric
limit-cycle support an OFD $({\footnotesize \text{ORTHOGONALITY}
\Rightarrow \text{STEADY-STATE SYMMETRY}})$. Please note that the
above statement is a \underline{strict one-way} relation\footnote{
the arrow in the above equation only goes one-way (orthogonality
implies symmetry) as also stated in the introduction. This follows
as we have only shown that $h \in \mathbf{H}_C$ maps a symmetric
limit-cycle, $\xi \in \mathcal{SYM}_n$, of the OFD PNF-OSC,
$\mathbf{o} \in \mathbf{O}^{\perp}$, on $\mathbb{U}$, onto an OFD
oscillator $\mathbf{q} \in \mathbf{O}^{\perp}$, on $\mathbb{W}_s$,
with an equally symmetric limit-cycle $\gamma \in
\mathcal{SYM}_n$. Nowhere, in the discussion herein, is it ever
claimed that an oscillator cannot have a symmetric limit-cycle and
\underline{not} be a member of $\mathbf{O}^{\perp}$. In-fact it is
easily proven, by way of simple counter-examples (see
\cref{sec0:foot1}), that such a claim would be
false.\label{sec0x:foot1}}.
\end{enumerate}

Finally, it is proven in \cref{sec3a} that the analysis results,
carried out using the specific PNF-OSC model, are in-fact unique.
As discussed in \cref{sec3a} this follows straight from the fact
that, $\overset{c}{\sim}$, represents an equivalence relation.
\Cref{sec5} details a series of numerical experiment, spanning
several different oscillator circuits, all of which,
unequivocally, support the claim introduced in \cref{sec3:theo1}.

\subsection{Summary of methodology}

Consider \cref{sec2b:fig2} which shows the PNF-OSC parameterizing
an unspecified oscillator circuit; \emph{i.e.} the N-OSC. The
dynamics of the PNF-OSC circuit is illustrated in
\cref{sec2b:fig1} which shows transient orbits plus the
limit-cycle of this circuit. This oscillator circuit is designed
to be an OFD oscillator and we find that the limit-cycle is
symmetric. The N-OSC circuit in \cref{sec2b:fig2}, could be any
hyperbolic $n$-dimensional oscillator \emph{e.g.} one of the
circuits described in \cref{sec5}. The parametrization map, $h$,
transforming the PNF-OSC orbits into N-OSC orbits (and vice-versa
through $h^{-1}$) is guaranteed to exist and to be unique (see
\cref{sec2}). This transformation preserves invariant spaces and
hence maps the PNF-OSC limit-cycle into the equivalent N-OSC set
(see \cref{sec2b:fig2} caption).
\par
A conjugation-map, $h$, relating two OFD oscillators, must belong
to the special subset of maps $\mathbf{H}_C$
(\cref{sec2c:prop1,sec2c:prop2}). Hence, only a special type of
map can transform between two OFD oscillators. The PNF-OSC is, by
design, an OFD oscillator and it follows that the N-OSC
(righthand-side of \cref{sec2b:fig2}) will be an OFD as well if,
and only if, $h$ belongs to the set $\mathbf{H}_C$. So now it
becomes clear why the PNF-OSC was chosen as an OFD oscillator. We
are trying to describe the subset of all N-OSC circuits which are
also OFD oscillators. This then means that both the PNF-OSC and
N-OSC will be OFD circuits which then automatically restricts the
possible conjugation/parametrization-map candidates from any
smooth map, $h\in \mathbf{H}$, to the much smaller subset $h\in
\mathbf{H}_C \subset \mathbf{H}$. This analysis trick proves
fruitful as we are able to partially characterize this smaller
subset of maps (see \cref{sec3}). Specifically, the theory shows
that a map in $\mathbf{H}_C$ transform the limit-cycles in
\cref{sec2b:fig2} in-terms of simple linear orthogonal map (scalar
scaling + rotation + inversion). This map preserves symmetry and
the N-OSC limit-cycle will hence also be symmetric. This last
statement is the main result of this paper, \emph{i.e.}
\cref{sec3:theo1}, which represents a strict one-way relation (see
\cref{sec0x:foot1}.)
\par
The question of uniqueness of this result, derived using the
specific PNF-OSC model, becomes important to discuss. Why can we
not just chose some other OFD model, call it XNF-OSC, perhaps even
with a non-symmetric limit-cycle and then produce a completely
different conclusion? As explained in \cref{sec3a}, the uniqueness
of \cref{sec3:theo1} is saved by fact that conjugation
(parametrization) represents an equivalence relation. Basically,
if the XNF-OSC parameterizes the N-OSC then it must also
parameterize the PNF-OSC and vice-versa (via inverse map) which
then implies that the XNF-OSC (an OFD oscillator) must have a
symmetric limit-cycle (see above discussion). Hence, whether we
parameterize the N-OSC in \cref{sec2b:fig2} using the PNF-OSC, the
XNF-OSC, or any other possible OFD template, is irrelevant. We
will always arrive at the same result in \cref{sec3:theo1}.

\section{Basic theory}

\label{sec1}

The oscillator state is governed by a $n$-dimensional vector-field
$f : \mathbb{R}^n \to \mathbb{R}^n$ generating a set of $n$
coupled ordinary differential equations (ODE) $\dot{x} = f(x)$
with $x(t) : \mathbb{R} \to \mathbb{R}^n$ being the
$n$-dimensional state-vector parameterized by time $t$. The
solution, corresponding to the initial condition $x(0)= x_0$ is
written $x(t) = \phi_t(x_0)$ with $\phi_t :
\mathbb{R}\times\mathbb{R}^n \to \mathbb{R}^n$ known as the
\emph{flow}. The oscillator DC-point (quiescent start-up point),
$x_q$, is a zero-point of the vector field $f(x_q)=0$ and hence a
\emph{fixed-point} of the flow $\phi_t(x_q) = x_{q}$ for all $t$.
The oscillator ODE generates a hyperbolic $1$-dimensional
attractor, $\gamma$, known as a \emph{limit cycle}. The oscillator
PSS, $x_s(t+T) = x_s(T)$ is then a $T$-periodic orbit,
corresponding to an initial condition in $\gamma$, \emph{i.e.}
$x_s(t) = \psi_t(x_0)$ with $x_0 \in \gamma$. Herein, the ODE
description is assumed time-normalized using the time-scale $\tau
= 2\pi ( t / T )$ which results in a $2\pi$ periodic limit-cycle
solution $x_s(\tau + 2\pi ) = x_s(\tau)$ corresponding to an
oscillator frequency $\omega_0 = 2\pi/T = 1$. The term, oscillator
$\mathbf{q}$, refers herein to the solution pair $\mathbf{q}
=(\phi_{\tau},\gamma)$.

\subsection{The stable manifold \& isochrone foliation}

\label{sec1a}

\begin{figure}[!h]
\begin{center}
\includegraphics[scale=1.0]{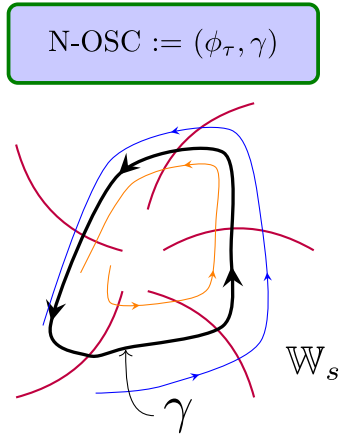}
\end{center}
\caption{A hyperbolic $n$-dimensional oscillator $\mathbf{q} =
(\phi_{\tau},\gamma)$, \emph{i.e.} a N-OSC, and the isochrone
foliation of the stable manifold $\mathbb{W}_s(\gamma)$. The
purple curves represent $4$ leaves in the foliation of
$\mathbb{W}_s(\gamma)$ (see \cref{sec1a:eq1}). The blue and orange
curves represent orbits of $\phi_{\tau}$ which approach $\gamma$
asymptotically with time.} \label{sec1a:fig1}
\end{figure}

We consider a $n$-dimensional, hyperbolically stable, oscillator
(N-OSC), $\mathbf{q} = (\phi_{\tau},\gamma)$, with PSS $x_s(\tau)
= x_s(\tau + 2\pi)$ (time-normalized). The oscillator stable
manifold defines the connected subset of $\mathbb{R}^n$ containing
all initial conditions which converge towards the oscillator
limit-cycle asymptotically with time $\lim_{\tau\rightarrow
\infty}\vert \phi_{\tau}(x_0) - \gamma \vert \rightarrow 0$ for \,
$x_0 \in \mathbb{W}_s$. The text herein considers electrical
oscillators with a single DC/start-up point, $x_q$, assumed to lie
at the origin\footnote{one can always assume a fixed-point at
the origin since a constant translation leaves the dynamics
unaffected. Hence an oscillator with DC/start-up point at $x_{q}=
\alpha \in \mathbb{R}^n$ is the same oscillator, dynamically
speaking, as an oscillator with fixed point at the origin $x_{q} =
0$. \label{sec1a:foot1}} which produces a stable manifold
consisting of $\mathbb{R}^n$ minus the origin \emph{i.e.}
$\mathbb{W}_s = \mathbb{R}^n \setminus\{ 0 \}$.
\par
It is a well-established
fact\cite{winfree1967,guckenheimer1975,djurhuus2009} that an open
subset $\mathbb{W}_s(\gamma) \subset \mathbb{W}_s$ of this
manifold, known as the stable manifold of $\gamma$ (dimension
$\geq n-2$), can be foliated by continuum of $(n-1)$-dimensional
hyper-surfaces known as \emph{isochrones} \emph{i.e.} equal
time/phase sets

\begin{equation}
\mathbb{W}_s(\gamma) = \underset{\eta \in [0,2\pi)}{\bigcup}
\bigl\{ \mathcal{I}_{\eta}\bigr\} \label{sec1a:eq1}
\end{equation}

where each leaf, $\mathcal{I}_{\eta}$, of this foliation contain
points in $\mathbb{W}_s(\gamma)$ with asymptotic oscillator phase
$\eta$. The topics discussed above are illustrated schematically
in \cref{sec1a:fig1}.

\subsection{Floquet theory \& the oscillator tangent-bundle}

\label{sec1b}

\begin{figure}[!h]
\begin{center}
\includegraphics[scale=1.0]{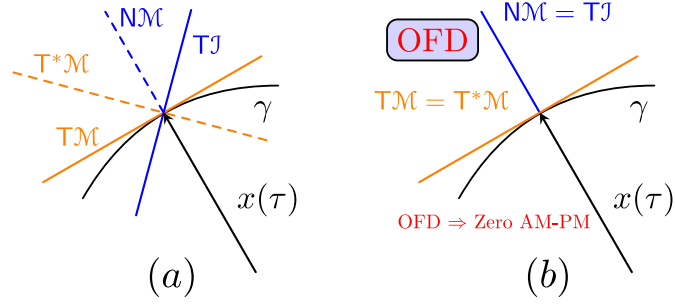}
\end{center}
\caption{\textbf{(a)} : the decomposition of the tangent-bundle
$\mathsf{T}_{\gamma}\mathbb{R}^n = \mathsf{T}\mathcal{M} \oplus
\mathsf{T}\mathcal{I}$ (limit-cycle/isochrone tangent spaces,solid line) and corresponding dual
tangent-bundle $\mathsf{T}_{\gamma}^*\mathbb{R}^n =
\mathsf{T}^*\mathcal{M} \oplus \mathsf{N}\mathcal{M}$ (dashed
line). This construct, based on the geometric notion of an
isochrone foliation of $\mathbb{W}_s(\gamma)$ (see
\cref{sec1a:fig1}). \textbf{(b)} : the orthogonal Floquet
decomposition (OFD) $\mathsf{T}_{\gamma}\mathbb{R}^n \in
\mathbf{B}^{\perp}$. For an orthogonal system the distinction
between basis and dual basis is erased.  } \label{sec1b:fig1}
\end{figure}

The N-OSC, $\mathbf{q} = (\phi_{\tau},\gamma)$, introduced above,
is hyperbolic and a unique corresponding set of $n$ (dual) Floquet vectors
$\{u_i(\tau),v_i(\tau)\}_{i=1}^n : \mathbb{R} \to \mathbb{R}^n$,
are then known to
exist~\cite{kartner1990,demir2000,traversa2011,djurhuus2009}.
These objects obey the bi-orthogonality condition (ODE systems)
$(v_i(\tau) , u_j(\tau) ) = \delta_{ij} \, \text{for} \, i =
1,2,\cdots n$ where $(\cdot,\cdot)$ designates the inner Euclidian
product and $\delta_{ij}$ is the Kroenecker delta-function.
\par
The fundamental-matrix map (F-MATRIX), $d\phi_{\tau} : \mathbb{R}
\times \mathbb{R}^n \to \mathbb{R}^{n\times n}$, is derived by
linearizing the full flow, $\phi_{\tau}$, around the limit-cycle,
$\gamma$. This map describes the oscillator linear-response (LR)
and hence governs orbits generated by weak perturbations around
$x_s(\tau) \in \gamma$. It can be decomposed as $d\phi_{\tau} =
\sum_{j=1}^n \exp(\mu_j \tau)u_j(\tau)v_j^{\top}(0)$ where $\mu_i$
is the time-normalized characteristic Floquet exponent
corresponding to $i$th (dual) Floquet
vectors\cite{demir2000,djurhuus2009}. It follows directly that
$u_i(0) {=} u_i(2\pi)$ is an eigenvector of the special F-MATRIX,
$d\phi_{2\pi} = \sum_{j=1}^n \exp(2\pi \mu_j
)u_j(2\pi)v_j^{\top}(0)$, known as the Monodromy Matrix
(M-MATRIX), with corresponding eigenvalue $\lambda_i =
\exp(2\pi\mu_i)$ being the $i$th Floquet characteristic
multiplier. Oscillator stability demands $\Re\{\mu_i\} \leq 0$
(real part $\leq 0$) for all $i$. Henceforth, the triple
$\{u_i(\tau),v_i(\tau),\mu_i\}$, or individual constituents of
this triple, will be referred to as the $i$th (Floquet) mode. The
dynamics of the circuits considered herein are real and modes
hence must appear in conjugate pairs (\emph{i.e.} $u_i =
u_{i+1}^*$) with \emph{real modes} (\emph{i.e.} modes with zero
imaginary parts) appearing as singles. For
oscillator solutions a special \emph{phase-mode},
$\{u_1(\tau),v_1(\tau),\mu_1 = 0\}$, is known to exist, describing
the neutrally stable dynamics tangential to the limit-cycle
$\gamma$. The associated dual vector, $v_1$, is known as the
\emph{perturbation-projection-vector} (PPV) \cite{demir2000,srivastava2007}. It
can be readily shown that $u_1$ is proportional to $\dot{x}_s$ and
we fix $u_1(\tau) =
\dot{x}_s(\tau)$~\cite{kartner1990,demir2000,traversa2011,djurhuus2009}.

\subsubsection{The oscillator tangent-bundle}

For a hyperbolic solution, the
Floquet collection $\{u_i(\eta)\}$, constitute a complete set and
hence form a basis for $\mathbb{R}^n$. The origin of this
vector-space is the PSS point, $x_s(\eta)$. This translated
vector-space is known as the \emph{tangent-space}
$\mathsf{T}_{\eta}\mathbb{R}^n = \text{span}\{ u_1(\eta) ,
u_2(\eta) , \cdots u_n(\eta) \}$ at $x_s(\eta)$. The disjoint
union of all these tangent-spaces, one for each point on the PSS
$x_s(\tau) = x_s(\tau+2\pi) \in \gamma$, is then is known as the
(oscillator) \emph{tangent-bundle}

\begin{equation}
\mathsf{T}_{\gamma}\mathbb{R}^n = \underset{\eta \in
[0,2\pi)}{\bigcup} \mathsf{T}_{\eta}\mathbb{R}^n = \text{span}\{
u_1(\eta) , u_2(\eta) , \cdots u_n(\eta) \} \label{sec1b:eq1}
\end{equation}

Both the limit-cycle and isochrone foliation (see
\cref{sec1a:fig1}) are invariant sets
under the flow  which then directly implies the following bundle
representation $\mathsf{T}_{\gamma}\mathbb{R}^n =
\mathsf{T}\mathcal{M} \oplus \mathsf{T}\mathcal{I}$ and from \cite{djurhuus2009}, $\mathsf{T}\mathcal{M} =
\cup \mathsf{T}_{\eta}\mathcal{M} = \text{span}\{ u_1(\eta)\}$,
$\mathsf{T}\mathcal{I} = \cup
\mathsf{T}_{\eta}\mathcal{I} = \text{span}\{ u_2(\eta) ,
u_3(\eta) , \cdots u_n(\eta) \}$ with $\eta \in [0,2\pi)$. The special mode, $u_1(\tau)$, spanning
$\mathsf{T}\mathcal{M}$, is the so-called \emph{phase-mode}
whereas the set $\{u_{i>1}(\tau)\}$ are the $(n-1)$
\emph{amplitude-modes}. Repeating the above discussion for the
adjoint F-MATRIX  we can derive the \emph{oscillator dual
tangent-bundle} $\mathsf{T}^{*}_{\gamma}\mathbb{R}^n {=}
\text{span}\{ v_1(\eta) , v_2(\eta) , \cdots v_n(\eta) \}$ which
is decomposed as $\mathsf{T}_{\gamma}^*\mathbb{R}^n =
\mathsf{T}^*\mathcal{M} \oplus \mathsf{N}\mathcal{M}$ where
$\mathsf{T}^*\mathcal{M} = \text{span}\{ v_1(\tau)\}$ (the PPV
bundle) and $\mathsf{N}\mathcal{M} = \text{span}\{ v_2(\tau) ,
v_3(\tau) , \cdots v_n(\tau) \}$ \cite{djurhuus2009}.
\par
Let $\mathbf{B}$ be the set of all possible tangent-bundles of the
form in \cref{sec1b:eq1}. The subset, $\mathbf{B}^{\perp} \subset
\mathbf{B}$, then hold all orthogonal bundle decompositions, or
more specifically, all \emph{orthogonal Floquet decomposition} (OFD)

\begin{equation}
\mathbf{B}^{\perp} = \bigl\{ \mathsf{T}_{\gamma}\mathbb{R}^n \in
\mathbf{B} : \,\, u_k(\tau) \perp u_j(\tau) \, \, \text{for} \,\,
k\neq j \,\, , \,  \forall \tau \bigr\} \label{sec1b:eq5a}
\end{equation}

where $\perp$ and $\forall$ are mathematical symbols for
\emph{orthogonal} and \emph{for all}. Let $\mathbf{O}$ be the set
of all hyperbolic, stable oscillators. The subset
$\mathbf{O}^{\perp} \subset \mathbf{O}$, then contain the special
\emph{OFD oscillators}

\begin{equation}
\mathbf{O}^{\perp} = \bigl\{ \, \mathbf{q} = (\phi_{\tau},\gamma)
\in \mathbf{O} \, : \, \mathsf{T}_{\gamma}\mathbb{R}^n \in
\mathbf{B}^{\perp} \, \bigr\} \label{sec1b:eq5}
\end{equation}

thus $\mathbf{O}^{\perp}$ hold all the oscillators with a tangent
bundle in $\mathbf{B}^{\perp}$. The concepts discussed here are
illustrated in \cref{sec1b:fig1}.
\par
Consider a noise signal $\varrho(\tau) : \mathbb{R} \to
\mathbb{R}^n$, perturbing the oscillator PSS. In order to
facilitate a discussion of AM-PM noise conversion, an orthogonal
frame $\{e_j(\tau)\}_{i=1}^n$, moving over (\emph{i.e.} with
origin at) $x_s(\tau)$, is introduced with $e_1(\tau)$ being
tangent to $\gamma$ implying $e_1(\tau) \parallel u_1(\tau)$. The
remaining $n-1$ components $\{e_j(\tau)\}_{i=2}^n$, all correspond
to directions orthogonal to $\gamma$ and can then viewed as
contra-variant versions of the set $\{v_j(\tau)\}_{i=2}^n$
spanning $\mathbb{NM}$ (see \cref{sec1b:fig1}). The noise-signal
is then decomposed as $\varrho(\tau) = \text{PM-noise + AM-noise}
= a_1(\tau)e_1(\tau) + \sum_{i = 2}^n a_i(\tau) e_i(\tau)$. Note,
that we are free to decompose the noise using any frame. The
calculated phase-noise spectrum is unaffected. The Floquet frame
is only inherent to the oscillator LR itself not to any perturbing
signal. From the bi-orthogonality condition $(v_i(\tau),u_j(\tau))
= \delta_{ij}$, discussed above, together with the OFD oscillator
definition in \cref{sec1b:eq5a,sec1b:eq5}, the phase-mode $u_1$
and PPV $v_1$ must be parallel, $u_1(\tau)
\parallel v_1(\tau)$, for all $\tau$ (see also \cref{sec1b:fig1}.b).
In this OFD scenario, the PPV hence only collects noise along $u_1
\propto \dot{x}_s$, and there will hence be no integration of
AM-noise (as defined herein) into the phase of the oscillator
which is the definition of zero AM-PM noise conversion.

\begin{remark}
An OFD oscillator, $\mathbf{q} \in \mathbf{O}^{\perp}$, has zero
AM-PM noise conversion.
\end{remark}

\section{Topological conjugate oscillators}

\label{sec2}

We consider the two $n$-dimensional open sets $\mathbb{U} =
\mathbb{W}_s = \mathbb{R}^n\backslash \{ 0\}$ parameterized by
coordinates $y = (y_1, y_2, \cdots, y_n)$ and $x = (x_1, x_2,
\cdots, x_n)$, respectively. Here $\mathbb{W}_s$ is the stable
manifold for the N-OSC, $\mathbf{q} = (\phi_{\tau},\gamma)$,
generated by the ODE $\dot{x} = f(x)$. The new domain, $\mathbb{U}$, known herein as
the \emph{parametrization manifold}, is the stable-manifold
of the \emph{normal-form oscillator} (NF-OSC), $\mathbf{o} =
(\psi_{\tau},\xi)$, generated by the ODE $\dot{y} = g(y)$.
\par
The theory of topological conjugate flows
\cite{wiggins1990,kuznetsov2013,frankel2011,cabre2003,cabre2003_2,cabre2005},
loosely speaking, describes a scenario wherein the NF-OSC,
$\mathbf{o}$, on $\mathbb{U}$, is used to \emph{parameterize} the
N-OSC, $\mathbf{q}$, on $\mathbb{W}_s$. The idea is to construct a
canonical (simple) model on $\mathbb{U}$ and
then perform all analysis on this simplified representation.
This concept is akin to applying
a basis change in standard linear-algebra in-order to simplify
solution procedure. The setup discussed here is fully symmetric,
($\mathbb{U} = \mathbb{W}_s$), which implies that $\mathbf{q}$ on
$\mathbb{W}_s$ can equally well be said to parameterizes
$\mathbf{o}$ on $\mathbb{U}$. However, in-order not to
unnecessarily complicate or confuse matters we stick with the
picture developed above (for now) wherein $\mathbb{U}$ contains
our simple normal-form oscillator (NF-OSC) $\mathbf{o} =
(\psi_{\tau},\xi)$ whereas $\mathbb{W}_s$ contains the
complex/real-life oscillator $\mathbf{q} = (\phi_{\tau},\gamma)$
(N-OSC) which we seek to model/parameterize (see also
\cref{sec2a:fig1} at this point).
\par
The theory of topological conjugacy, on which
our analysis herein is based, is an established branch of
dynamical systems theory with research stretching back
decades, if not
centuries\cite{kuznetsov2013,wiggins1990,frankel2011,huguet2012,huguet2013,huguet2018,cabre2003,cabre2003_2,cabre2005}.

\subsection{Basic Theory}

\label{sec2a}

Let $h : \mathbb{R}^n \to \mathbb{R}^n$ be a smooth transformation
between points on $\mathbb{U}$ and $\mathbb{W}_s$, respectively,
\emph{i.e.} $h(y) = x$, and let $\delta y,\delta x$ be vectors in
the respective tangent-spaces,
$\mathbb{T}_{u}\mathbb{R}^{n},\mathbb{T}_{x}\mathbb{R}^{n}$ (see
\cref{sec1b}). We then have

\begin{equation}
h(y) = x \Rightarrow dh(y) \delta y = \delta x \label{sec2a:eq1}
\end{equation}

where $dh : \mathbb{R}^n \to \mathbb{R}^{n\times n}$ is the
Jacobian matrix of the map $h$. \Cref{sec2a:eq1} is an axiomatic
(self-evident/explanatory) identity which follows directly from
standard theory of smooth maps
\cite{frankel2011,wiggins1990,kuznetsov2013}; \emph{i.e.} $h$ maps
between points whereas the Jacobian $dh$ maps between vectors in
tangent-spaces at these points. The vector-field, $g$, on
$\mathbb{U}$ thus maps to a \emph{topological equivalent}
vector-field, $f$, on $\mathbb{W}_s$ as $dh(y)g(y) = f(x) =
f(h(y))$. By varying the map $h$ every possible equivalent field
on $\mathbb{W}_s$ can be thus constructed. We see that have
effectively \emph{parameterized} $f$ on $\mathbb{W}_s$ in-terms of
the model field, $g$, on $\mathbb{U}$. Integrating this relation,
\emph{w.r.t.} $\tau$, gives $h(\psi_{\tau}(y)) = \phi_{\tau}(h(y))
\Leftrightarrow \psi_{\tau} = h^{-1} \circ \phi_{\tau} \circ h$
where the relation between fields and flows (\emph{i.e.}
$\partial(\phi_{\tau}(x))/\partial \tau {=} f(x)$, see text in
\cref{sec1}) was used and $s \circ p \equiv s(p)$ denotes
composition of functions $s$ and $p$. Assuming
time-parametrization is preserved, which will be the case for
hyperbolic systems, the flows are said to be
\emph{topological-conjugate} \cite{kuznetsov2013,wiggins1990}. The
conjugation relation is represented herein by the operator
$\sim_h$, or simply $\sim$, and the conjugation of oscillators
$\mathbf{o} = (\psi_{\tau},\xi)$ and $\mathbf{q} =
(\phi_{\tau},\gamma)$ is then written

\begin{equation}
\psi \sim_h \phi := \psi_{\tau} = h^{-1} \circ \phi_{\tau} \circ h
\label{sec2a:eq2}
\end{equation}

where below we also use the notation $\mathbf{o} \sim \mathbf{q}$.
\Cref{sec2a:eq2} directly implies a similar conjugacy of the
corresponding iterated maps $\psi^{(k)} \sim \phi^{(k)}$ where
$\alpha^{(k)} = \alpha \circ \alpha \cdots \circ \alpha$ ($k$
times). The operator $\sim$ is an
equivalence relation\footnote{
to show this we need to prove reflexivity, symmetry and transitivity  \cite{wiggins1990,kuznetsov2013}. The
operator is reflexive $\kappa \sim_h \kappa$ for $h=id$ (the
identity map). Then $\kappa \sim_h \theta \Leftrightarrow \theta
\sim_g \kappa$ for $g = h^{-1}$ proves symmetry. Let $\theta =
f\circ \kappa \circ f^{-1}$ and $\theta = g^{-1} \circ \sigma
\circ g$. Then $f\circ \kappa \circ f^{-1} = g^{-1} \circ \sigma
\circ g \Leftrightarrow (g \circ f) \circ \kappa = \sigma \circ (
g \circ f)$ and $\kappa = ( g \circ f )^{-1} \circ \sigma \circ (
g \circ f)$. Thus $\kappa \sim_f \theta \wedge \theta \sim_g
\sigma \Rightarrow \kappa \sim_h \sigma$ with $h = f \circ g$
proving transitivity.\label{sec2a:foot1}}.

\begin{figure}[!h]
\begin{center}
\includegraphics[scale=1.0]{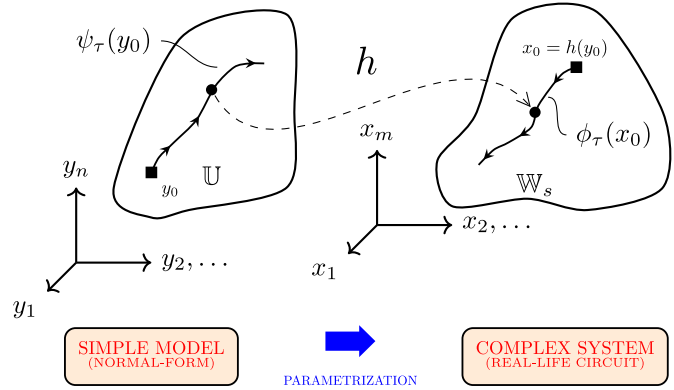}
\end{center}
\caption{The conjugation, $h$, maps orbits $\psi_{\tau}(u_0)$,
with initial condition (point at $\blacksquare$) $y_0$ in the
parameter space $\mathbb{U}$, to orbits $\phi_{\tau}(x_0)$ on
$\mathbb{W}_s$, corresponding to the initial condition $x_0 =
h(u_0)$.} \label{sec2a:fig1}
\end{figure}

Given the orbit, $\psi_{\tau}(y_0)$, on $\mathbb{U}$, with initial
condition $\psi_0(y_0) = y_0$, the conjugate flow,
$\phi_{\tau}(x_0)$, will correspond to the orbit with initial
condition $\phi_{0}(x_0) = x_0 = h(y_0)$ on $\mathbb{W}_s$ and we
can write \cref{sec2a:eq2} as

\begin{equation}
h \circ \psi_{\tau}(y_0) = \phi_{\tau}(x_0) \label{sec2a:eq3}
\end{equation}

which shows that $h$ maps orbits in $\mathbb{U}$ onto orbits in
$\mathbb{W}_s$ while keeping the time parametrization.
\Cref{sec2a:fig1} gives a schematic illustration of the topics
discussed here. Consider an \emph{invariant-set} $Q$ in
$\mathbb{U}$ \emph{i.e.} $\psi_{\tau}(Q) \subset Q$ for all
$\tau$. From the conjugation relation \cref{sec2a:eq3} it then
follows that $\phi_{\tau}(h(Q)) = \phi_{\tau}(S) = h \circ
\psi_{\tau}(Q) \subset h(Q) = S$, where $S = h(Q)$, and

\begin{equation}
\psi_{\tau}(Q) \subset Q \Leftrightarrow \phi_{\tau}(S) \subset S
\label{sec2a:eq4}
\end{equation}

where the left arrow follows from considering the inverse
transformation $h^{-1}$ (\emph{i.e.} the map from $\mathbb{W}_s$
to $\mathbb{U}$). An invariant set, under $\psi_{\tau}$ on
$\mathbb{U}$, $Q$, thus corresponds to an invariant set, $S =
h(Q)$, under the conjugated flow, $\phi_{\tau}$, on
$\mathbb{W}_s$; and vice-versa. Repeating this analysis for the
iterated map version of \cref{sec2a:eq2} it follows that invariant
sets under the iterated map are also preserved. The oscillator
limit-cycles $\xi,\gamma$, are invariant sets of the flows
$\psi_{\tau},\phi_{\tau}$ on $\mathbb{U}$ and $\mathbb{W}_s$,
respectively, and from \cref{sec2a:eq4}

\begin{equation}
h(\xi) = \gamma \label{sec2a:eq5}
\end{equation}

Likewise, the leaves,of the isochrone
foliation on $\mathbb{U}(\xi)$ are an invariant
of $2\pi$-iterated map, $\psi_{2\pi}$.

\subsection{Choosing a specific NF-OSC template model : the PNF-OSC}

\label{sec2b}

\begin{figure}
\begin{center}
\includegraphics[scale=1.0]{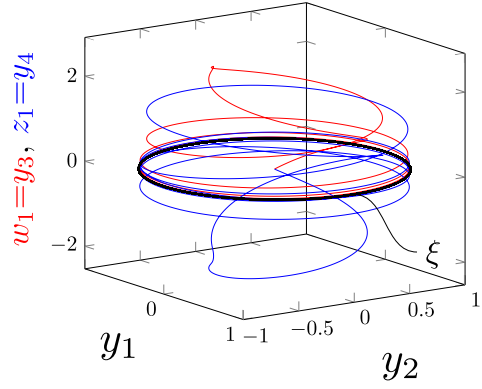}
\end{center}
\caption{Transient orbits, and limit-cycle $\xi$, of the PNF-OSC
system defined in  \cref{sec2b:eq1}. The example system considered
here is $5$-dimensional with $m=1,k=1$ ($1$ additional real + $1$
complex mode) described in-terms of coordinates $\bar{y} =
(r,\phi,w_1,z_1,z_2) \in \mathbb{R}^5$ where $w_1 = y_3$, $z_1 =
y_4$ and $z_2 = y_5$ (see text). Parameters are fixed as
$(\mu,\beta_1,\sigma_1,\nu_1) = (0.5,0.2,0.1,0,3)$. Two cross
sections are shown where $(y_1,y_2,w_1=y_3)$ (red curve) shows the
effects of the additional real mode whereas $(y_1,y_2,z_1=y_4)$
(blue curve) which shows the effects of the complex mode.}
\label{sec2b:fig1}
\end{figure}

We consider the $n$-dimensional parameter manifold $\mathbb{U}$
indexed by the Cartesian coordinate set $y = (y_1, y_2, \cdots,
y_n) \in \mathbb{R}^n$ (see \cref{sec2} and \cref{sec2a:fig1}).
Let $r = \sqrt{y_1^2 + y_2^2}$, $\phi = \arctan(y_2/y_1)$ be the
polar coordinates indexing the $y_1,y_2$ plane and consider the
transformed coordinates $\bar{y} = (\phi , r , w ,z) \in
\mathbb{R}^n$ where sub-coordinate vectors $w \in \mathbb{R}^m$
and $z \in \mathbb{R}^{2k}$ contain the remaining $n{-}2$
$y$-coordinates\footnote{here coordinate sub-vectors $w,z$ simply
represent place-holders for the remaining $n-2$ coordinates
$(y_3,y_4,\cdots,y_n)$ with the constraint that sub-vector $z$ is
of even dimension $2k$. Hence $w = (y_3,y_4,\cdots y_{m+2}) \in
\mathbb{R}^m$ and $z = (y_{m+3},y_{m+4}, \cdots y_{m+2 +2k}) \in
\mathbb{R}^{2k}$ where the dimensions of these sub-vectors are
chosen such that $n-2 = m + 2k \Rightarrow n = m + 2 +
2k$.\label{sec2b:foot2}}. As $\bar{y}$ is an $n$-dimensional
coordinate system the dimension of these sub-vectors are
constrained through $n = 2 + m +2k$ (see \cref{sec2b:foot2}). As
$n\geq 2$ is assumed, it always possible to find two numbers $m,k
\geq 0$ such that this constraint is upheld.

\begin{proposition}
The new coordinates $\bar{y} = ( \phi , r , w, z ) \in
\mathbb{R}^n$, indexing the parameter manifold $\mathbb{U}$, are
orthogonal. \label{sec2b:prop1}
\end{proposition}
\begin{proof}

The polar coordinates $(\phi,r)$, indexing the $y_1,y_2$ plane,
represent an orthogonal coordinate system in this plane
\cite{frankel2011}. The remaining coordinates $(y_3,\cdots , y_n )
\in \mathbb{R}^{n-2}$ are orthogonal to the $y_1,y_2$ plane and
furthermore comprised of Cartesian coordinate functions; which are
of-course, by definition, orthogonal. However, coordinate
sub-vectors $w,z$ are simply containers holding this remaining set
of Cartesian $y$-coordinates (see text above and
\cref{sec2b:foot2}) and it hence follows trivially that $\bar{y}=
(\phi , r , w , z ) = (\phi , r , y_3, y_4,\cdots y_n ) \in
\mathbb{R}^n$ is an orthogonal coordinate system.
\end{proof}

In these new coordinates, the dynamics of the NF-OSC on
$\mathbb{U}$, is modelled in-terms of the autonomous ODE, $\dot{y}
= g(y)$, of the form

\begin{equation}
\begin{aligned}
 &\dot{\phi} =  1 \\
 &\dot{r}  = \mu r(1-r) \\
 &\dot{w}_{i} =  -\beta_iw_i \quad   &i = 1,2,\cdots , m\\
 &\dot{z}_{2i-1} = -\sigma_i z_{2i-1} + \nu_iz_{2i} \quad  &i = 1,2,\cdots , k \\
 &\dot{z}_{2i} = -\sigma_i z_{2i} - \nu_iz_{2i-1} \quad   &i = 1,2,\cdots , k \label{sec2b:eq1}
\end{aligned}
\end{equation}

with $\mu,\alpha,\beta_i,\sigma\in \mathbb{R}^+$ are a collection
of positive real parameters whereas $\nu_i \in
\mathbb{R}\backslash \{0\}$ is a non-zero real parameter. From
inspection, it follows that the system in \cref{sec2b:eq1}
generates a single stable limit-cycle set

\begin{equation}
\xi = \bigl\{ \bar{y} \in \mathbb{U} \, \, : \, \,
(\phi,r,\{w,z\}) = [0;2\pi)\times 1 \times \{\boldsymbol{0}\}\bigr
\} \label{sec2b:eq2}
\end{equation}

with $\{\boldsymbol{0}\} = 0\times 0 \times 0 \times \cdots \times
0$ (m+2k times) and \cref{sec2b:eq2} hence simply describes a unit
circle in the $y_1,y_2$ plane. Henceforth, the NF-OSC, $\mathbf{o}
= (\psi_{\tau},\xi)$, where $\psi_{\tau}$ is the flow on
$\mathbb{U}$ generated by integrating the ODE in \cref{sec2b:eq1},
will be known as the PNF-OSC (polar normal-form).

\begin{proposition}
\label{sec2b:prop2}the PNF-OSC system, defined in-terms of the ODE
in \cref{sec2b:eq1}, belong to the OFD oscillator class,
$\mathbf{o} \in \mathbf{O}^{\perp}$.
\end{proposition}
\begin{proof}
From \cref{sec1b:eq1} the bundle $\mathsf{T}_{\gamma}\mathbb{R}^n$
is spanned by the Floquet vectors $(u_1(\tau),u_2(\tau),\cdots ,
u_n(\tau))$. From \cref{app1:eq3} in appendix \ref{app1} we have
$\mathsf{T}_{\xi}\mathbb{R}^n = \text{span}(\hat{\phi} , \hat{r} ,
\{ \hat{w}_i\}_{i=1}^m , \{ \hat{z}_{2i} \pm j\hat{z}_{2i-1}
\}_{i=1}^{k})$ with the notation, $\hat{x}$, representing the
coordinate-vector corresponding to coordinate function, $x$. From
\cref{sec2b:prop1} the coordinate system $\bar{y} = (\phi , r , w
, z)$ is orthogonal which implies that the corresponding
coordinate vectors $\hat{\phi},\hat{r}, \hat{w}_i$ \emph{etc.} are
orthogonal. By definition (see \cref{sec1b:eq5a})
$\mathsf{T}_{\xi}\mathbb{R}^n \in \mathbf{B}^{\perp}$ and from the
definition in \cref{sec1b:eq5} we get $\mathbf{o} \in
\mathbf{O}^{\perp}$.
\end{proof}

The dynamics corresponding to the $w$ and $z$ coordinate sets
generate $m$ real and $k$ imaginary stable Floquet modes. The
orbits of an example 5-dimensional PNF-OSC system ($m = k = 1$)
were calculated by numerically integrating \cref{sec2b:eq1} and
the resulting curves are plotted in \cref{sec2b:fig1} together
with the limit-cycle, $\xi$, defined in \cref{sec2b:eq2}.

\begin{figure}
\begin{center}
\includegraphics[scale=1.0]{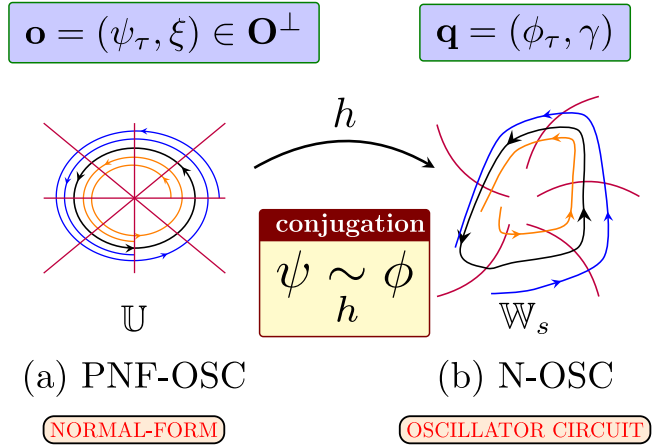}
\end{center}
\caption{ \textbf{(a)} The PNF-OSC oscillator $\mathbf{o} =
(\psi_{\tau},\xi) \in \mathbf{O}^{\perp}$ with stable
(parametrization) manifold $\mathbb{U}$. \textbf{(b)} the unspecified N-OSC
oscillator $\mathbf{q} = (\phi_{\tau},\gamma)$ with stable
manifold $\mathbb{W}_s$. The PNF-OSC provides a parametrization of
the N-OSC in $\mathbb{W}_s$ in-terms of the conjugation map $h$.
This parametrization is faithful, mapping the limit-cycle in
$\mathbb{U}$, $\xi$, onto the corresponding limit-cycle in
$\mathbb{W}_s$, $\gamma$, and the individual isochrone foliation
leaves (purple curves) in $\mathbb{U}$ onto the corresponding sets
in $\mathbb{W}_s$ (see discussion in \cref{sec1a}).}
\label{sec2b:fig2}
\end{figure}

\subsection{The set $\mathbf{H}_C$ and operator $\overset{c}{\sim}_h$}

\label{sec2c}

We seek to use parametrization, $\mathbf{o} \sim \mathbf{q}$,
in-order to identify all N-OSC oscillators $\mathbf{q} =
(\phi_{\tau},\gamma)$ on $\mathbb{W}_s$ which belong to the class
$\mathbf{O}^{\perp}$. From \cref{sec2b:prop2}, $\mathbf{o}\in
\mathbf{O}^{\perp}$, and this in-turn places certain limitations on $h\bigl|_{\xi}$ ($h$ restricted to $\xi$)

\begin{proposition}
\label{sec2c:prop1}
Let $\mathbf{o} \sim_h \mathbf{q}$ where $\mathbf{o} =
(\psi_{\tau},\xi) \in \mathbf{O}^{\perp}$ is the PNF-OSC defined
in \cref{sec2b:eq1,sec2b:eq2}. Then $\mathbf{q} \in
\mathbf{O}^{\perp} \Leftrightarrow h\bigl|_{\xi} \text{is
conformal (angle preserving)}$.
\end{proposition}
\begin{proof}

From \cref{sec2a:eq1} it follows that at every point of the
PNF-OSC PSS orbit $y_s(\tau) \in \xi$, the Jacobian $dh$ maps
tangent-spaces in $\mathbb{T}_{\xi}\mathbb{R}^n$ onto
tangent-spaces $\mathbb{T}_{\gamma}\mathbb{R}^n$ of the conjugated
orbit $x_s(\tau) \in \gamma$. From \cref{sec2b:prop2},
$\mathbb{T}_{\xi}\mathbb{R}^n \in \mathbf{B}^{\perp}$, and
tangent-spaces along $\xi$ are all spanned by orthogonal
basis-vectors. The N-OSC tangent-bundle (see \cref{sec1b}) will be
orthogonal, $\mathbb{T}_{\gamma}\mathbb{R}^n \in
\mathbf{B}^{\perp}$, if, and only if, the Jacobian is
angle-preserving (conformal) at all points of $\xi$. This implies
that $dh(y_s(\tau)) \in \text{CO}(n)$ must hold for all $\tau$
where we let $\text{CO}(n)$ be the set of all $n\times n$
conformal matrices. By definition, $h\bigl|_{\xi}$ is then
conformal.
\end{proof}

Let $\mathbf{H}$ be the set of all conjugation maps. The subset
$\mathbf{H}_{\text{C}} \subset \mathbf{H}$, and the conjugation
operator $\overset{c}{\sim}$, are then defined as

\begin{equation}
\begin{aligned}
\mathbf{H}_{\text{C}} &= \{ h \in \mathbf{H} : h\bigl|_{\xi}
\text{is
conformal} \} \\
\overset{c}{\sim}_h & := \, \sim_h \,\, \wedge \,\, h \in
\mathbf{H}_{\text{C}}
\end{aligned}
\label{sec2c:eq1}
\end{equation}

and $\overset{c}{\sim}$ is simply $\sim$ with the extra condition
that the conjugation map, $h$, belongs to $\mathbf{H}_{\text{C}}$
(conformal restriction). It can be shown that $\overset{c}{\sim}$
is an equivalence operator\footnote{simply apply
\cref{sec2a:foot1} plus the fact the composition preserves the
conformal property meaning that $s \circ p$ is conformal if, and
only if, both $s$ and $p$ are conformal. \label{sec2c:foot1}}. We
can then re-state \cref{sec2c:prop1}

\begin{proposition}[\cref{sec2c:prop1} re-stated]
\label{sec2c:prop2} Let $\mathbf{o} \in \mathbf{O}^{\perp}$ be the
PNF-OSC on $\mathbb{U}$. Then $\mathbf{q} \in \mathbf{O}^{\perp}
\Leftrightarrow \mathbf{o} \overset{c}{\sim} \mathbf{q}$.
\end{proposition}
\begin{proof}
Follows directly from \cref{sec2c:prop1} and the definitions in
\cref{sec2c:eq1}.
\end{proof}

Let us briefly explain this result. Assuming the PNF-OSC,
$\mathbf{o} = (\psi_{\tau},\xi)$, on $\mathbb{U}$, is used as a
parametrization template. \Cref{sec2c:prop2} then says that,
in-order for $\mathbf{q} = (\phi_{\tau},\gamma)$ on $\mathbb{W}_s$
to be an OFD oscillator, there must exist a conjugation map $h$,
with a conformal restriction on $\xi$ (\emph{i.e.} $h \in
\mathbf{H}_{\text{C}}$) such that $\mathbf{o} \overset{c}{\sim}_h
\mathbf{q}$ (and \underline{not} just the standard $\mathbf{o}
\sim_h \mathbf{q}$).

\section{Main result : OFD and PSS symmetry correlation}

\label{sec3}

The conjugation relation \cref{sec2a:eq2} is valid at every point
of domains $\mathbb{U}$ and $\mathbb{W}_s$, respectively. This, by
definition, then implies that it also holds on all open sets
contained in these spaces\cite{wiggins1990,kuznetsov2013}. We
consider the following open neighborhoods $\Gamma_{\xi} = \{y \in
\mathbb{U}, \epsilon \in \mathbb{R}^+ : \Vert y - \xi \Vert \leq
\epsilon\} \subset \mathbb{U}$ and $\Gamma_{\gamma} = \{x \in
\mathbb{W}_s, \varepsilon \in \mathbb{R}^+ : \Vert x - \gamma
\Vert \leq \varepsilon\} \subset \mathbb{W}_s$ which describe
tubular open sets enclosing the limit-cycles $\xi,\gamma$.

\begin{proposition}

Consider the (local) conjugation map $h : \Gamma_{\xi} \to
\Gamma_{\gamma}$. If $h \in \mathbf{H}_C$, with $\mathbf{H}_C
\subset \mathbf{H}$ defined in \cref{sec2c:eq1}, then this map
must have the following representation on $\Gamma_{\xi}$

\begin{equation}
h(y) = \nu(y) + H(y) \label{sec3:eq1}
\end{equation}

where $\nu : \Gamma_{\xi} \to \Gamma_{\gamma}$ is a conformal map
which restricts to $h$ on $\xi$, $h(\xi) \equiv  h\bigl|_{\xi} =
\nu\bigl|_{\xi} \equiv \nu(\xi)$ and where $H : \Gamma_{\xi} \to
\Gamma_{\gamma}$ is some unspecified map with $H(\xi) = 0$.

\label{sec3:prop1}

\end{proposition}

\begin{proof}

For every conformal restriction, $h\bigl|_{\xi}$, at-least one
conformal map, $\nu : \Gamma_{\xi} \to \Gamma_{\gamma}$, must
exist; \emph{i.e.} simply continue the power-series expansion of
the restriction $h\bigl|_{\xi}$, around $\xi$, in every possible
way that keeps $\nu$ conformal.  Let $H = h - \nu$ denote the
residual. By definition this residual is non-conformal on $\xi$
since all conformal contributions are contained in $\nu$. It then
follows directly\footnote{ Here $H\bigl|_{\xi} = H(\xi)$ is
non-conformal and $dH\bigl|_{\xi} \notin \text{CO}(n)$. Then $h =
\nu + H \in \mathbf{H}_C$ if, and only if, $dH$ is zero everywhere
on $\xi$. But this implies that $H$ must be constant on $\xi$,
meaning $H(\xi) = c$ for some scalar $c \in \mathbb{R}^n$. Here
the map, $H(\xi) = c$, represent a constant translation of the
full conjugation map, $h(y) = \nu(y) + H(y)$, which is irrelevant
as the dynamics invariant under constant translations (see
\cref{sec1} and \cref{sec1a:foot1}). Hence $c$ can be any value
w/o changing the outcome of the analysis and we choose $c=0$ which
keeps the singular/fixed-point of the oscillator at the
origin.\label{sec3:foot2}} that  $H(\xi) = 0$.
\end{proof}

The following statement discuss the possible forms the conformal
map $\nu$, in \cref{sec3:prop1}, can take

\begin{proposition}
\label{sec3:prop2}

Any conformal \emph{conjugation map} $\nu : \Gamma_{\xi} \to
\Gamma_{\gamma}$, must have the form

\begin{equation}
\nu(y) = \rho A y \label{sec3:eq1}
\end{equation}

where $\rho \in \mathbb{R} \setminus \{0\}$ is a non-zero real
value and $A \in O(n)$ is an orthogonal matrix.

\end{proposition}
\begin{proof}
For $n>2$ this is a direct consequence of Liouville's theorem
(1850)
\cite{hartman1947,flanders1966,jacobowitz1991,kuhnel2007,iwaniec1998}.
This theorem states that all conformal maps on an open region of
$\mathbb{R}^n$, with $n>2$, must be M{\"o}bius transformations of
the form $\nu(y) = b + \rho A(y - a )/|y - a |^s$, with $a,b \in
\mathbb{R}^n$, $\rho\in \mathbb{R}$, $A \in O(n)$ and the integer
exponent $s$ is either $0$ or $2$. We first consider the case
$s=2$. This map will have a singularity at $y = a$. From the
discussion in \cref{sec2a}, the conjugation-map transform orbits
into orbits (see \cref{sec2a:eq3}). This singularity would hence
imply that orbits around $y=a$ would be mapped to orbits around
infinity; clearly not a possibility. Hence we must have $s=0$
which implies a map $\nu(y) = d + \rho A y$ where $d = b - \rho
Aa$ is a real parameter. Both oscillators are assumed to have
fixed-points at the origin (see \cref{sec1,sec2} and \cref{sec1a:foot1,sec3:foot2}) and we hence
must have $\nu(0) = 0$ which implies $d = 0$ and we reached
\cref{sec3:eq1}. For the special planar case, $n=2$, $\nu$ is a
bi-holomorphic map on the annuli $\Gamma_{\xi}$. Schottkys Theorem
(1877)\cite{astala2008} states that the map must have the form
$k(z) = az^{\pm 1}$ where $a\in \mathbb{C} \setminus \{0\}$.
Again, since singularities are not allowed and we must have $k(0)
= 0$ this implies the map $k(z) = az$. \Cref{sec3:eq1} is then
reached by transforming from complex to real coordinates in the
plane.
\end{proof}

At this point we define the set of rotational symmetric curves
centered at the origin and with radius $r$, on $\mathbb{W}_s$

\begin{equation}
\mathcal{SYM}_n = \bigl\{ x \in \mathbb{R}^n , \exists r \in
\mathbb{R}^+ \,\, : \,\, \lVert x \rVert - r = 0 \bigr\}
 \label{sec3:eq2}
\end{equation}

From \cref{sec2b:eq2}, the PNF-OSC limit-cycle belongs to this
set, $\xi \in \mathcal{SYM}_n$. This fact allow us to state the
main result of this paper

\begin{theorem}

\label{sec3:theo1}

Let $\mathbf{q} = (\phi_{\tau},\gamma) \in \mathbf{O}$ be any
given oscillator on $\mathbb{W}_s$. Then

\begin{equation}
\mathbf{q} \in \mathbf{O}^{\perp} \Rightarrow \gamma \in
\mathcal{SYM}_n
\end{equation}

\end{theorem}
\begin{proof}

From \cref{sec2c:prop2} we must have $\mathbf{o} \overset{c}{\sim}
\mathbf{q}$ where $\mathbf{o} = (\psi_{\tau},\xi)$ is the PNF-OSC
described in \cref{sec2b}. This implies that the conjugation map,
$h$, must belong to the set, $\mathbf{H}_\text{C}$, defined in
\cref{sec2c:eq1}. From \cref{sec3:prop1}, this implies a
conjugation map of the form $h(y) = \nu(y) + H(y)$ with $h(\xi) =
\nu(\xi)$. \Cref{sec3:prop2} and \cref{sec2a:eq5} then yield
$\gamma = h(\xi) = \nu(\xi) = \rho A y$ , with $y \in \xi$. This
describes a linear rotation + scaling of $\xi$ which implies $\xi
\in \mathcal{SYM}_n \Rightarrow \gamma = h(\xi) \in
\mathcal{SYM}_n$, as linear rotation + scaling preserves symmetry
of $\xi \in \mathcal{SYM}_n$ (\emph{i.e.} maps $\mathcal{SYM}_n$
into $\mathcal{SYM}_n$).
\end{proof}

Firstly, it is important to note that the result in
\cref{sec3:theo1} represents a one-way implication. In other
words, the OFD property \emph{implies} PSS symmetry. No reverse
relation exists as discussed, at length, in the introduction;
\emph{i.e.} symmetry does \underline{not} imply an OFD. Secondly,
the result, implicitly, also hold for symmetric limits-sets on
$\mathbb{W}_s$ with center away from the origin even-though
$\mathcal{SYM}_n$ in \cref{sec3:eq2} seem to include this
restriction\footnote{this is because a non-zero center for
$\gamma$, on $\mathbb{W}_s$, can always brought back to the origin
through a simple linear translation which leaves the dynamics
un-changed. This linear translation, as was explained in
\cref{sec1} (see also \cref{sec1a:foot1}), is assumed a-priori to
analysis. Hence the location of the limit-cycle center is
irrelevant; only symmetry is important.\label{sec3:foot1}}.
\Cref{sec3:theo1} explains why OFD's generally are not observed in
real-life oscillator systems where non-linear device-models make
it impossible to attain perfect PSS symmetry. The methodology
developed in the previous sections, leading to the main result in
\cref{sec3:theo1}, is given the SYM-OFD calling handle; signifying
the (one-way) relation between PSS symmetry and LR OFD described
in \cref{sec3:theo1}.

\subsection{Discussion of the result in \cref{sec3:theo1}}

\label{sec3a}

It may seem that the result in \cref{sec3:theo1} is of limited scope
since it relies on the specific choice of the PNF-OSC normal-form
model. Luckily, because $\overset{c}{\sim}$ is an equivalence
relation (see \cref{sec2c:eq1} and \cref{sec2c:foot1}), this turns
out not to be an issue. Assume we had chosen some other OFD NF-OSC
model, call it XNF-OSC, $\mathbf{x} =
(\theta_{\tau},\sigma) \in \mathbf{O}^{\perp}$ with stable manifold $\mathbb{X} = \mathbb{R}^n\backslash \{ 0\}$. Since the PNF-OSC
is an OFD oscillator, $\mathbf{o} \in \mathbf{O}^{\perp}$, it
follows from \cref{sec2c:prop2} that a map, $k : \mathbb{U} \to \mathbb{X} \in \mathbf{H}_C$, must exist
such that $\psi_{\tau}
\overset{c}{\sim}_k \theta_{\tau}$. Equivalently, if the N-OSC $\mathbf{q} =
(\phi_{\tau},\gamma)$ is an OFD oscillator, $\theta_{\tau}
\overset{c}{\sim}_f \phi_{\tau}$. Since $\overset{c}{\sim}$ is an
equivalence relation (see \cref{sec2c:eq1} and
\cref{sec2c:foot1,sec2a:foot1}) it follows from transitivity

\begin{equation}
\psi_{\tau} \overset{c}{\sim}_k  \theta_{\tau} \wedge
\theta_{\tau} \overset{c}{\sim}_f \phi_{\tau} \Rightarrow
\psi_{\tau} \overset{c}{\sim}_h \phi_{\tau} \label{sec3a:eq1}
\end{equation}

where $h$ is the composition $h = k \circ f \equiv k(f)$. The
statement in \cref{sec3:theo1} is hence independent of the choice
of normal-form (NF-OSC) and hence unique. Or said another way :
\cref{sec3a:eq1} shows that the XNF-OSC, $\mathbf{x}$, and the PNF-OSC, $\mathbf{o}$,
give the same OFD classification of the N-OSC, $\mathbf{q}$.
\par
The range of an equivalence relation, \emph{i.e.} the types of
oscillators which can be parameterized, is limited by the
topological invariants. Each set of invariant values defines a
unique \emph{equivalence class} which is preserved under
conjugation \cite{wiggins1990,kuznetsov2013}. The invariants do
not change under conjugation and it follows that a oscillators in
one equivalence class (defined in-terms of one sets of invariants)
cannot parameterize a system in a different class; as this would
correspond to a different set of invariants. For hyperbolic
oscillators, on identical domains $\mathbb{U} = \mathbb{W}_s$, the
only invariant of importance is the M-MATRIX stability measure
characterized in-terms of the NF-OSC M-MATRIX
eigenvalue-spectrum\footnote{ expanding the conjugation relation
\cref{sec2a:eq2}, around $\xi$, in a power-series, using
\cref{sec2a:eq5} and excluding higher-order terms gives the
expression $dh(u_{\tau}) \circ d\psi_{\tau}(u_0) =
d\phi_{\tau}\circ dh(u_0)$. Setting $\tau = 2\pi$ and using the
notation for the M-MATRIX introduced in \cref{sec1b} \emph{i.e.}
$\Psi = d\psi_{2\pi}$, $\Phi = d\phi_{2\pi}$, we get $\Phi =
dh(u_0) \circ \Psi \circ dh(u_0)^{-1}$ where $u_0 = u_{2\pi}$ for
$u_0,u_{2\pi}\in \xi$, was used. This is a \emph{similarity
relation} and the M-MATRIX eigenvalues are hence invariant under
conjugation.}, $\text{spec}\{\Psi\}$. From \cref{sec2b:prop2} and appendix \ref{app1}, the
PNF-OSC was specifically designed such that all possible M-MATRIX
eigenvalue-spectra can be modelled, \emph{i.e} any combination of
$\{\lambda_i\}_{i=1}^{n}$, inside the unit-circle on the complex
plane ($\vert \lambda_i \vert \leq 1$ for stable oscillators), and
this model is hence able to parameterize all oscillators on
$\mathbb{W}_s$.
\par
The results derived herein are thus fully independent of the
specific choice NF-OSC model. The PNF-OSC model was only chosen
in-order to facilitate the analysis leading to \cref{sec3:theo1}.
We could easily have chosen any other equivalent model and would
have arrived at the exact same result. The equivalence topics
discussed herein are well-known concepts developed within the
field of topological conjugation theory; a mature and established
branch of modern dynamical systems theory. The SYM-OFD model, and
specifically \cref{sec3:theo1}, is hence based on an extremely
rigorous foundation of theoretical research which stretches back
several decades \cite{wiggins1990,kuznetsov2013,frankel2011,
huguet2012,huguet2013,huguet2018,cabre2003,cabre2003_2,cabre2005,hartman1947,flanders1966,jacobowitz1991,kuhnel2007,iwaniec1998,astala2008}.

\section{Numerical experiments}

\label{sec5}

\begin{figure*}[t]
\centering
\begin{subfigure}[t]{.6\textwidth}
  \centering
  \raisebox{1.5cm}{\includegraphics[scale=1.0]{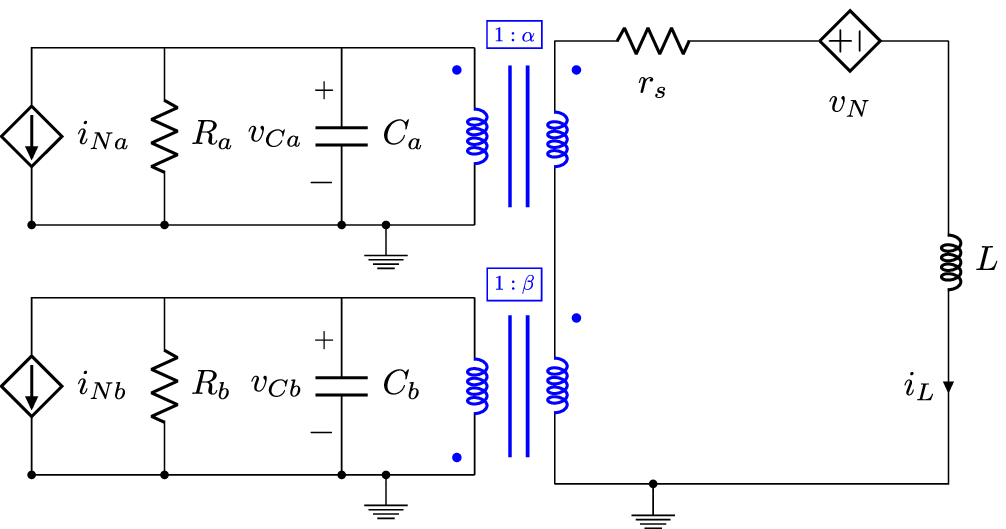}}
  \caption{TCR-OSC}
\end{subfigure}%
\hspace{0.25em}
\begin{subfigure}[t]{.35\textwidth}
  \centering
   \includegraphics[scale=0.6]{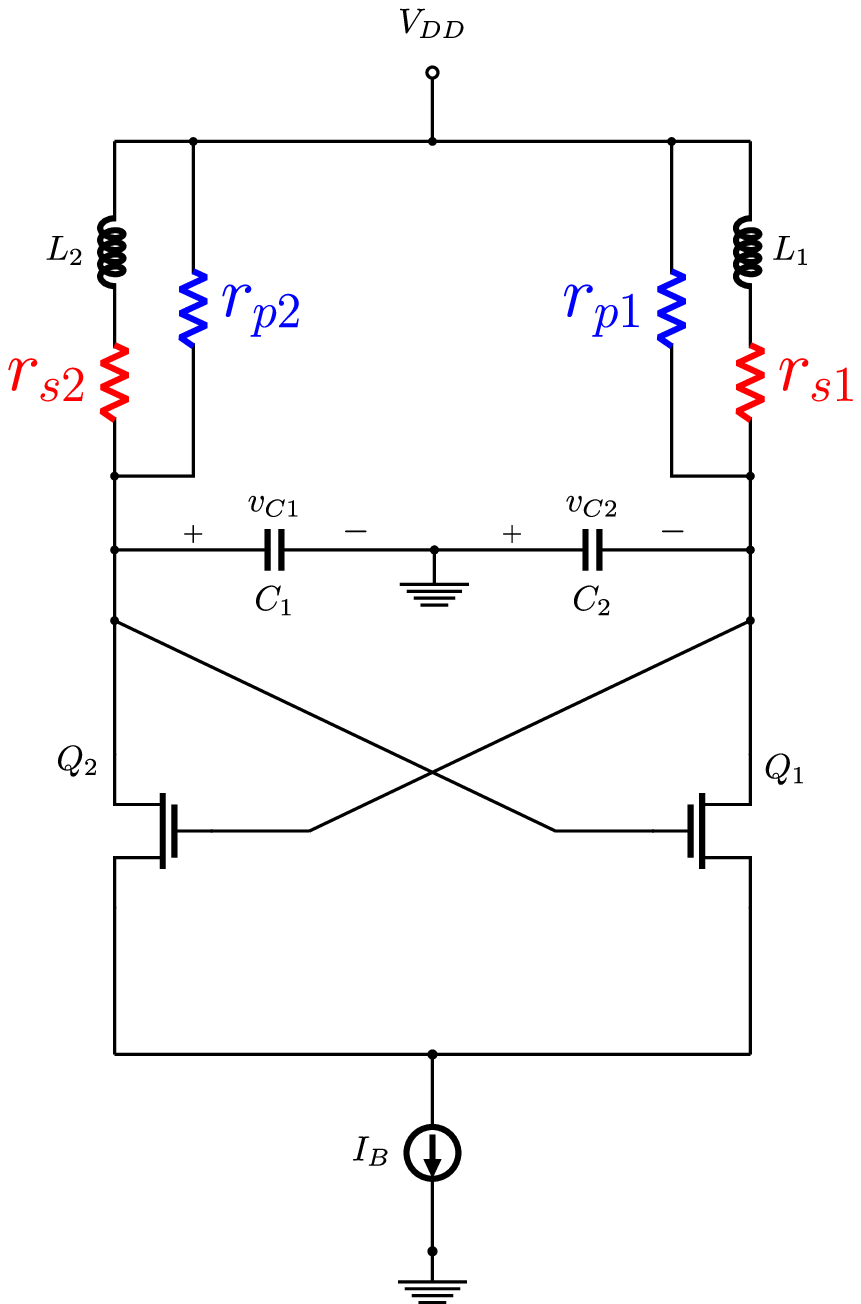}
    \caption{FET-OSC}
\end{subfigure}
\caption{(a) : transformer coupled resonator oscillator (TCR-OSC),
$(i_{Na},i_{Nb})$ are trans-conductors (nonlinear conductors) and
$v_N$ is a trans-impedance (nonlinear resistor). (b) MOSFET
oscillator (MOS-OSC) oscillator with both parallel (blue) and
series (red) resonator resistor options. } \label{sec5:fig1}
\end{figure*}

The transformers in the TCR-OSC circuit, shown in
\cref{sec5:fig1}.(a), are perfectly ideal with turn-ratios
$\alpha,\beta$ and with the lower transformer inducing a
$180^{\circ}$ phase shift. Time-normalization, $\tau = t\omega_0$,
is introduced with $\omega_0 = 1.0/\sqrt{LC_s}$ and $C_s =
C_{a}C_{b}/(C_a+C_b)$. The time-normalized equations of the
TCR-OSC circuit are then written

\begin{equation}
\begin{aligned}
\sigma_a \dot{y} &= -y/Q_a + \alpha_n v/z_0 + i_{Na}(y,v,w) \\
\dot{v} &= -v/Q_s + z_0(\alpha y - \beta w)  + v_{N}(y,v,w) \\
\sigma_b \dot{w} &= -w/Q_b + \beta_n v/z_0 + i_{Nb}(y,v,w)
\end{aligned}
\label{sec5:eq1}
\end{equation}

where $x = (y,v,w) = (v_{Ca}/z_0,i_Lz_0,v_{Cb}/z_0) \in
\mathbb{R}^3$ is the normalized $3$-D circuit state-vector with
$z_0 = \sqrt{L/C_s}$ and $\sigma_{a,b} = C_{a,b}/C_s$, $\alpha_n =
\alpha/(\alpha^2 + \beta^2)$, $\beta_n = \beta/(\alpha^2 +
\beta^2)$, $Q_{a,b} = \omega_0 C_a R_{a,b}$, $Q_s = \omega_0 L
/r_s$. The trans-conductors and trans-impedance controlled sources
(negative conductors/resistors) are defined in-terms of the
functions $i_{N(a\vert b)}(x) = (\alpha_n\vert \beta_n)z_0[\alpha
y - \beta w]f(x)$ and $v_N(x) = (v/z_0)f(x)$ with $f :
\mathbb{R}^3 \to \mathbb{R}$ being the real function $f(x) =
\bigl( p - \bigl(q_1[\alpha z_0 y - \beta z_0 w]^2 + q_2(v/z_0)^2
\bigr)$. These types of polynomial functions, discussed here, are
readily implemented electrically using \emph{e.g.} standard
operational amplifiers (op-amp) circuits, operational
trans-conductance amplifiers (OTA) or dedicated chips. We will not
dwell on this practical implementation issue here as this topic is
not critical for the discussion below. Unless stated otherwise the
circuit parameters are fixed as $\alpha=\beta = 1/\sqrt{2}$, $C_a
= C_b \Rightarrow \sigma_a = \sigma_b = 2$, $z_0 = 1 \Omega$,
$Q_a=Q_b=Q_s = 100$, $p=0.06\mathrm{A/V}$ , $q_1 = 0.003
\mathrm{A/V^{2}}$, $q_2 = 0.003 \mathrm{A^{-1}}$, $s =
1\mathrm{V/A}$.
\par
The FET-OSC circuit in \cref{sec5:fig1}.(b) is a differential pair
FET LC oscillator configuration. The circuit includes both an
option for parallel (blue) and series (red) resonator resistance
topology. In the simulations below the FET differential-pair
circuit is modelled in-terms of the trans-conductance function
$\Theta(\zeta) = I_B (2/\pi)\arctan( [\pi/2](G_m/I_B)\zeta)$ which
has been demonstrated to be a excellent approximation of the
actual device model\cite{pepe16}. Here $\zeta$ represents the
input voltage across the pair and $G_m = \sqrt{k_nI_B}$ is the
small-signal trans-conductance where $k_n = \mu_n C_{\text{ox}}
W_n/L_n$ with $\mu_n,C_{\text{ox}}, W_n$ and $L_n$ are the charge
mobility, oxide capacitance, gate width and length, respectively.
The normalized time-variable is given as, $\tau = t\omega_0$,
where $\omega_0 = 1.0/\sqrt{L_pC_s}$, $C_s = C_{1}C_{2}/(C_1+C_2)$
and $L_p = L_{1}L_{2}/(L_1+L_2)$. The time-normalized dynamic
equations, of the circuit in \cref{sec5:fig1}.(b),
are then written

\begin{equation}
\begin{aligned}
\sigma_1\dot{q} &= -\textcolor{blue}{q/Q_{p1}} - r/z_0 + \Theta(z_0[q-s]) \\
\kappa_1\dot{r} &=  z_0 q - \textcolor{red}{r/Q_{s1}}     \\
\sigma_2 \dot{s} &= -\textcolor{blue}{s/Q_{p2}} - u/z_0  -\Theta(z_0[q-s]) \\
\kappa_2 \dot{u} &= z_0s - \textcolor{red}{u/Q_{s2}}
\end{aligned}
\label{sec5:eq2}
\end{equation}

where $x = (q,r,s,u) = (v_{C1}/z_0,i_{L1}z_0,v_{C2}/z_0,i_{L2}z_0)
\in \mathbb{R}^4$ is the (normalized) $4$-D state vector with $z_0
= \sqrt{L_p/C_s}$. The blue/red contributions in \cref{sec5:eq2}
correspond to the blue/red parts of the FET-OSC circuit in
\cref{sec5:fig1}.(b), $Q_{p1,2} = \omega_0 C_{1,2}r_{p1,2}$,
$Q_{s1,2} = \omega_0 L_{1,2}/r_{s1,2}$, $\sigma_{1,2} =
C_{1,2}/C_s$ and $\kappa_{1,2} = L_{1,2}/L_p$. Unless otherwise
stated the parameters are fixed as $Q_{p1,2} = Q_{s1,2} = 100$,
$C_1 = C_2 \Rightarrow \sigma_{1,2} = 2.0$, $L_1 = L_2 \Rightarrow
\kappa_{1,2} = 2.0$, $I_B = 8 \mathrm{mA}$ and $k_n =
2\mathrm{AV^{-2}}$ (time-scaled value).

\subsection{The simulation measures}

\label{sec5a}

The N-OSC limit-cycle on $\mathbb{W}_s$ is written componentwise
as $\gamma(\tau) = (\gamma_1(\tau) , \gamma_2(\tau) , \cdots ,
\gamma_n(\tau))$ where $\gamma(\tau) = \gamma(\tau + 2\pi)$ (due
to time-normalization, see \cref{sec1}). Given the $2\pi$ periodic
function $\rho(\tau) = \sum_{i=1}^n \gamma_i^2(\tau)$ the
following two measures are introduced

\begin{equation}
\begin{aligned}
\Lambda &=  \Biggl| \frac{1}{2\pi} \int_{0}^{2\pi}
\biggl\{(\rho(\tau)/\rho_{\max}) - 1.0 \biggr\}d\tau \biggr| \\
\Upsilon &= \underset{i,j \in [1;n]}{\max} \bigl\{ \Theta_{ij}
\bigr\} \, ,\, \text{with} \\
\Theta_{ij} &= \underset{\tau \in
[0;2\pi]}{\max} \bigl| |\angle \{ u_{i}(\tau) , u_{j}(\tau) \}| -
\pi/2 \bigr| \label{sec5a:eq1}
\end{aligned}
\end{equation}

where $\rho_{\max}$ represents the maximal value of the function
$\rho$ on the interval $\tau \in [0;2\pi)$ and $u_{i}(\tau)$ is
the $i$th Floquet mode. From inspection\footnote{ from
\cref{sec3:eq2}, $\rho$ is a constant/scalar function for $\gamma
\in \mathcal{SYM}_n$. If $\rho$ is constant then \cref{sec5a:eq1}
gives $\Lambda = 0$ which implies $\gamma \in \mathcal{SYM}_n
\Rightarrow \Lambda = 0$. Likewise, since the integrand is
strictly negative or zero the only way $\Lambda$ can be zero is if
$\rho$ is a constant scalar function which directly implies
$\gamma \in \mathcal{SYM}_n$. Hence, $\Lambda = 0 \Rightarrow
\gamma \in \mathcal{SYM}_n$ and we have derived $\gamma \in
\mathcal{SYM}_n \Leftrightarrow \Lambda = 0$. $\Upsilon$ directly
measures the maximum deviation, away from $90^{\circ}$, of the various angles between modes and
hence the closeness of the solution to the OFD oscillator class $\mathbf{O}^{\perp}$ (see \cref{sec1b:eq5}).
\label{sec5a:foot1}} it should be clear that $\Lambda$ in measures
the closeness of the limit-cycle $\gamma$ on $\mathbb{W}_s$ to
members of the set $\mathcal{SYM}_n$ in \cref{sec3:eq2} whereas
$\Upsilon$ measures the closeness of a N-OSC solution, $\mathbf{q}
\in \mathbf{O}$, to a member of the special OFD oscillator class
$\mathbf{O}^{\perp}$ defined in \cref{sec1b:eq5}.

\begin{remark}
\label{sec5a:remark1} From \cref{sec5a:eq1} and
\cref{sec5a:foot1}, in-order for the result in \cref{sec3:theo1}
to be valid the relations $\Upsilon \rightarrow 0 \Rightarrow
\Lambda \rightarrow 0$ must be observed in all numerical
experiments. This is of-course equivalent to $\Upsilon \rightarrow
-\infty \Rightarrow \Lambda \rightarrow -\infty$ on the dB scale;
i.e. $10\log_{10}(\Lambda),10\log_{10}(\Upsilon)$.
\end{remark}

The predictions made in \cref{sec5a:remark1} are now tested in
\cref{sec5a:fig1} on a simple van-der-Pol oscillator. From the
insets in this figure, the PSS is seen to approach an elliptical
limit-set as $\epsilon \to 0$ for all values of $c_0$ except for
the value $c_0 = 1$. The exception, $c_0=1.0$, approaches a
rotational symmetric set $\gamma \in \mathcal{SYM}_2$ (see
\cref{sec3:eq2}). From \cref{sec5a:remark1}, we need to see that
$\Upsilon \rightarrow 0$, implies, $c_0 = 1$; exactly because
$\Lambda \rightarrow 0$, as $\epsilon \to 0$, only for $c_0 = 1$.
Inspecting \cref{sec5a:fig1} this is indeed the case.

\begin{figure}[t]
\begin{center}
\includegraphics[scale=1.0]{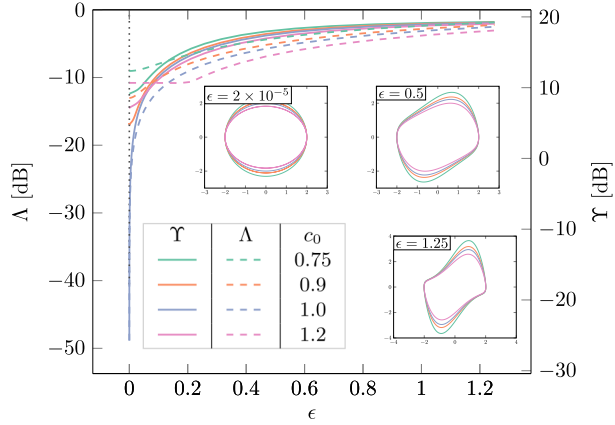}
\end{center}
\caption{ Simulations for the simple van-der-Pol oscillator,
$\dot{x} = c_0 y $, $\dot{y} = -x + \epsilon ( 1.0 - x^2 )y$. The
two measures $\Lambda$ and $\Upsilon$, defined in
\eqref{sec5a:eq1}, are plotted in the logarithmic dB scale
$10\log_{10}(\Lambda),10\log_{10}(\Upsilon)$) and insets of the
oscillator limit-cycles, at various parameter points (color-codes
match those in $\Upsilon,\Gamma$-tables), are included. }
\label{sec5a:fig1}
\end{figure}

\subsection{Simulation results}

\label{sec5b}

\begin{figure*}[t]
\begin{center}
\includegraphics[scale=0.75]{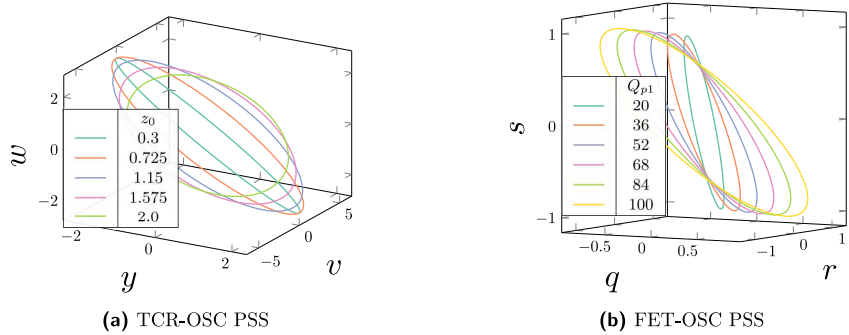}
\end{center}
\caption{ \textbf{(a)} : TCR-OSC PSS (see \cref{sec5:eq1},
\cref{sec5:fig1}.(a)) $z_0 : 0.3\to 2.0$. \textbf{(b)} : FET-OSC
PSS (see \cref{sec5:eq2}, \cref{sec5:fig1}.(b)) parallel resonator
(\textcolor{blue}{blue}), $Q_{p1} : 20\to 100$. }
\label{sec5b:fig1}
\end{figure*}

\begin{figure*}[t]
\begin{subfigure}[t]{1.0\textwidth}
\centering
\includegraphics[scale=0.75]{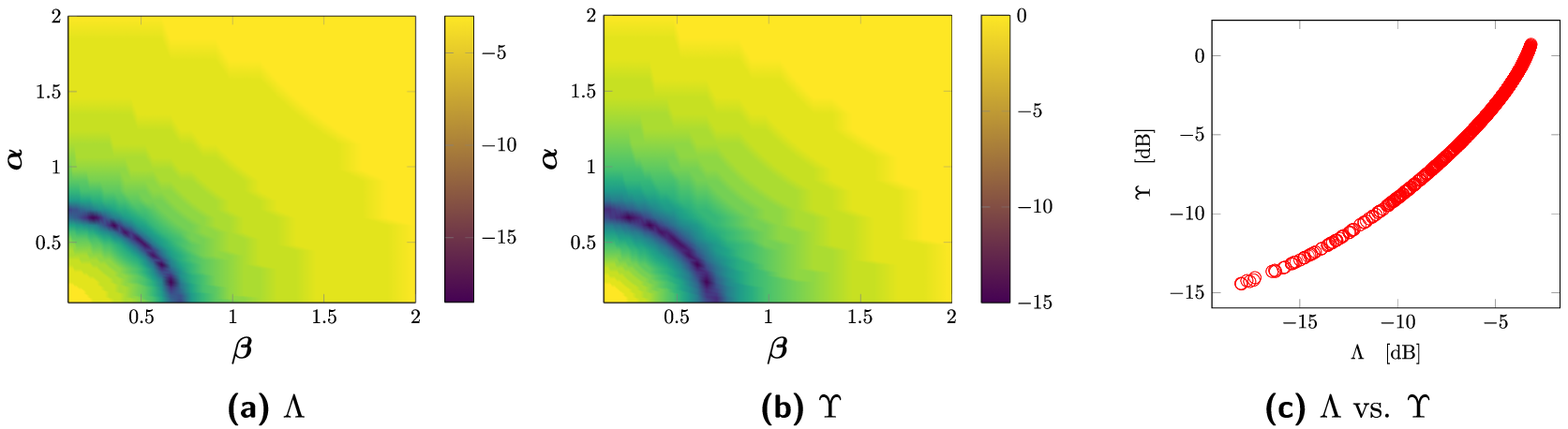}
\end{subfigure}\\
\begin{subfigure}[t]{1.0\textwidth}
\centering
\includegraphics[scale=0.75]{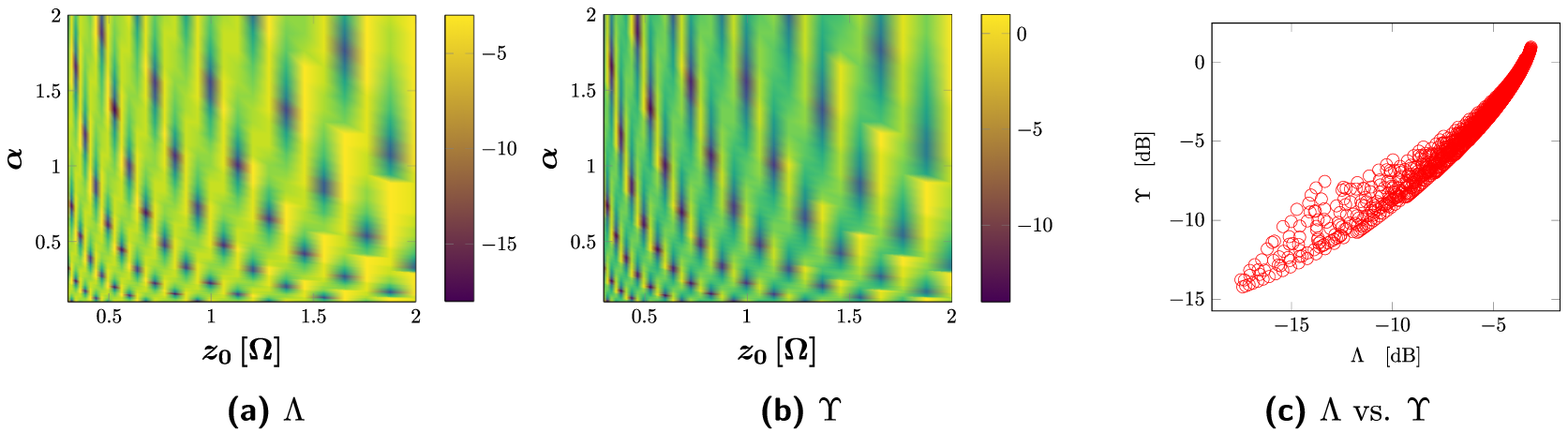}
\end{subfigure}
\caption{ Simulations for the TCR-OSC circuit in
\cref{sec5:eq1,sec5:fig1}.(a). \textbf{(a)-(b)} : 2D sweep of
measures $\Upsilon$ and $\Lambda$ in dB (\emph{i.e.}
$10\log_{10}(\Lambda), 10\log_{10}(\Upsilon)$) defined in
\cref{sec5a:eq1}, \textbf{(c)} : a scatter plot of the
points in \textbf{(a)} ($10\log_{10}(\Lambda)$) vs. the points in \textbf{(b)}
($10\log_{10}(\Upsilon)$). \textbf{Upper row : } sweep of parameters $\alpha,\beta \in
[0.1;2.0]\times [0.1;2.0]$. \textbf{Lower row :} sweep of parameters $\alpha,z_0 \in [0.1;2.0]\times [0.3;2.0]$.}
\label{sec5b:fig2}
\end{figure*}

\Cref{sec5b:fig1} plots the PSS orbits (limit-cycles)
corresponding to the oscillator circuits in \cref{sec5:fig1}. The
figures illustrate how the PSS symmetry can be
increased/descreased by varying the circuit parameters defined in
connection with \cref{sec5:eq1,sec5:eq2}. Once the PSS has been
calculated, the measure, $\Lambda$, is readily calculated from
\cref{sec5a:eq1}. Based on this PSS, the Floquet modes
$\{u_i(\tau),v_i(\tau),\mu_i\}$ are derived by integrating the LR
equations corresponding to this
solution\cite{kartner1990,demir2000,djurhuus2009} (see also
\cref{sec1b}) and the measure, $\Upsilon$, follows from
\cref{sec5a:eq1}.
\par
Considering the TRC-OSC circuit in \cref{sec5:fig1}.(a), the two
measures are calculated, as explained above, while sweeping
circuit parameters $(z_0,\alpha,\beta)$. \Cref{sec5b:fig2} (upper
row) depicts contour plots for the two measures $\Lambda,\Upsilon$
as functions of turn-ratios $\alpha,\beta$ swept in the interval
$[0.1;2.0]\times [0.1;2.0]$. Inspecting these two images one
immediate observation is that the pattern of the contours appear
very similar. This suspicion is then directly confirmed in
\cref{sec5b:fig2}.(c) which plots the points in these two figures
against each other as a scatter-plot format. It is strikingly
clear from this plot that the measures are very strongly
correlated. The figure furthermore reveals the relationship
$\Upsilon \to 0 \Rightarrow \Lambda \to 0$ ($-\infty$ on dB scale)
and from \cref{sec5a:remark1} this observation therefore directly
serves to validate the novel SYM-OFD model derived herein. In the
lower row of \cref{sec5b:fig2} these same simulations are repeated
but this time $\alpha,z_0$ are swept in the range $[0.1;2.0]\times
[0.3;5.0]$. The figure again shows the measures obeying the
predictions laid out in \cref{sec5a:remark1} and hence again serve
to validate the model. Finally, \cref{sec5b:fig3} shows two
equivalent simulation sweeps for the FET-OSC circuit in
\cref{sec5:fig1}.(b) with both figures serving to validate the
novel SYM-OFD model developed herein.

\begin{figure*}[t]
\begin{subfigure}[t]{1.0\textwidth}
\centering
\includegraphics[scale=0.75]{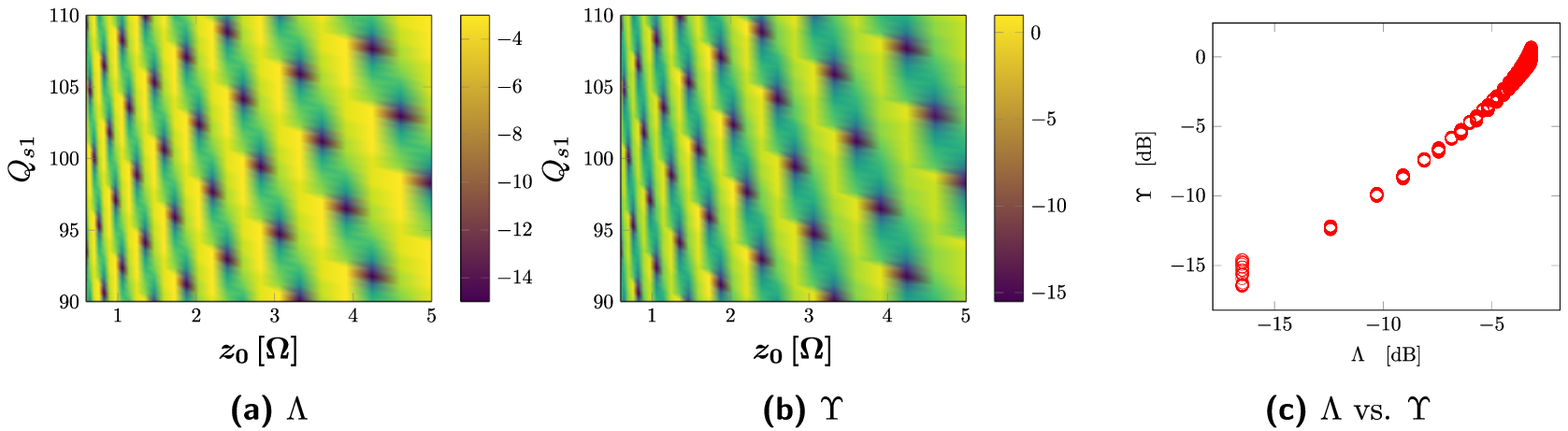}
\end{subfigure}\\
\begin{subfigure}[t]{1.0\textwidth}
\centering
\includegraphics[scale=0.75]{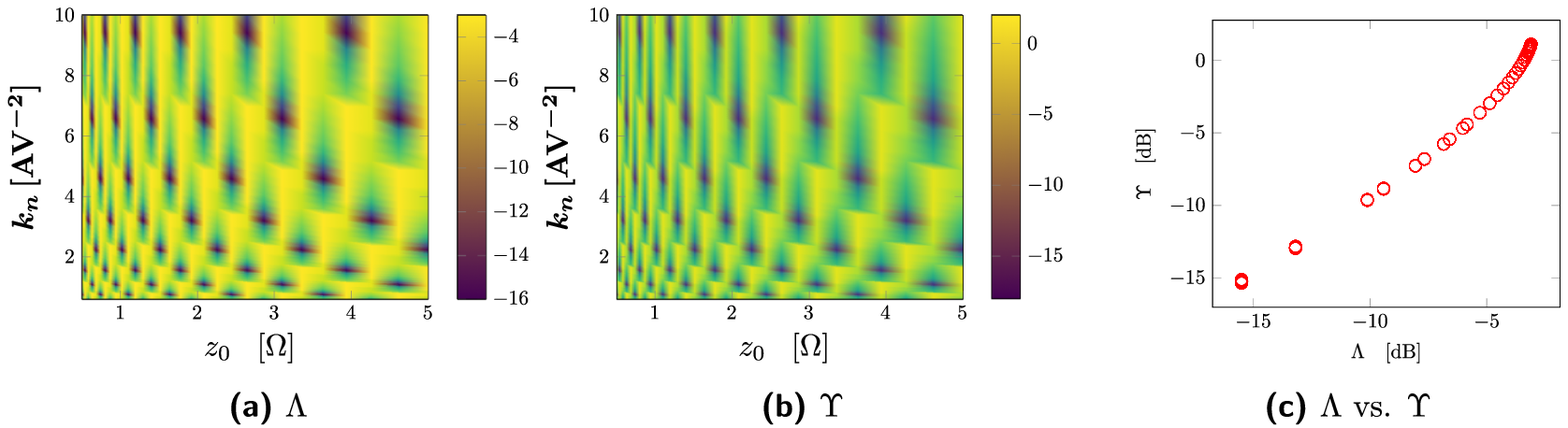}
\end{subfigure}
\caption{ Simulations for the FET-OSC circuit in
\cref{sec5:eq2,sec5:fig1}.(b). Figure content/symbols as in \cref{sec5b:fig2}.
\textbf{Upper row : } sweep of parameters $(Q_{s1},z_0) \in [90;110]\times [0.6;5.0]$. \textbf{Lower row :} sweep of parameters $z_0,k_n
\in [0.6;5]\times [0.5;10.0]$.}
\label{sec5b:fig3}
\end{figure*}

\subsection{Summary of Simulation Results}

The 13 simulations in
\cref{sec5a:fig1,sec5b:fig2,sec5b:fig3}, divided
into 3 figures, all together serve to validate the SYM-OFD model
developed herein. By varying the circuit parameters we can control
both the symmetry properties of the steady-state (see
\cref{sec5b:fig1}) and the Floquet decomposition of the LR (level
of orthogonality). Inspecting, \cref{sec5a:fig1} and the 2-D contour
plots + scatter plots in
\cref{sec5b:fig2,sec5b:fig3} it is clear that
these two properties are closely correlated, \emph{i.e.}
orthogonality increases $\Rightarrow$ symmetry increases; as was
first predicted in \cref{sec3:theo1}.
\par
It is important to note that this idea is completely novel. There
is absolutely nothing in the definition of symmetry/orthogonality
measures $\Lambda,\Upsilon$ in \cref{sec5a:eq1} that reveals that
these should be related in any way or form. Hence, the observation
described in \cref{sec5a:remark1} is completely unanticipated. It
cannot be reached by arguments based on empirical or
phenomenological reasoning but relies on the rigorous formulation
developed herein.

\section{Conclusion}
\label{sec6}

We prove that in order to achieve an orthogonal Floquet
decomposition of the oscillator linear-response, which implies
zero AM-PM noise conversion, an oscillator must produce a
rotationally symmetric steady-state solution. This novel link
connecting the configuration of a Floquet decomposition with
symmetry properties of the underlying steady-state is a novel
concept which, to our knowledge, is first discussed here for the
general case. Our results can be interpreted as a condition for
the leaves of the oscillator isochrone foliation to intersect the
limit-cycle set orthogonally. We performed a series of numerical
experiments on higher dimensional oscillators with all simulations
unequivocally supporting the predictions of the proposed model.

\appendix

\section{The PNF-OSC Floquet decomposition}

\label{app1}

The Jacobian matrix $J_{\xi} \in \mathbb{R}^{n\times n}$ is now
calculated as the derivative of the ODE in \cref{sec2b:eq1} at
every point on $\xi$ in \cref{sec2b:eq2}. Due to the perfect
symmetry of these equations it follows easily that this matrix is
constant (independent of coordinates) and has the block-diagonal
form $J_{\xi} = \text{diag}( 0 , -2\mu , W , Z)$ where $W \in
\mathbb{R}^{m\times m}$ is the $m$-dimensional constant diagonal
matrix $W = \text{diag}(-\beta_1,-\beta_2,\cdots , -\beta_m)$
whereas $Z \in \mathbb{R}^{2k\times 2k}$ is the constant
$2k$-dimensional block diagonal matrix $Z =
\text{diag}(Z_1,Z_2,\cdots , Z_k)$ with $Z_i = \bigl[\begin{smallmatrix} -\sigma_i & \nu_i\\
-\nu_i & -\sigma_i \end{smallmatrix}\bigr]$. Note that $W,Z$ are
the sub-Jacobian matrices corresponding to the $w,z$ equations in
\cref{sec2b:eq1}. The PNF-OSC F-MATRIX map $\Psi(t,s)$ (see
discussion \cref{sec1b}) is then calculated as the solution to the
Jacobian equation $\dot{\Psi}(t,s) = J_{\xi} \Psi(t,s)$. Given
that $J_{\xi}$ is constant, and block-diagonal, a closed form
solution is easily derived. The PNF-OSC Monodromy (M-MATRIX)
$\Psi_{2\pi}$ is then found from this F-MATRIX solution as
$\Psi_{2\pi} = \Psi(2\pi + s,s)$. Following this recipe the
PNF-OSC M-MATRIX is derived as

\begin{equation}
\begin{aligned}
\Psi_{2\pi} &= \text{diag}( 1 , -2\mu\pi , W_{2\pi} ,
Z_{2\pi}) \\
W_{2\pi} &=  \text{diag}( \exp(-2\beta_1\pi),\exp(-2\beta_2\pi)
, \cdots , \exp(-2\beta_m\pi)) \\
Z_{2\pi} &= \text{diag}(P\Lambda_1P^{-1},P\Lambda_2P^{-1},\cdots,
P\Lambda_kP^{-1}) \label{app1:eq1}
\end{aligned}
\end{equation}

where $\Lambda_i = \text{diag}( \exp(2\pi[-\sigma_i + j\nu_i]),
\exp(2\pi[-\sigma_i - j\nu_i])$ and $P =  \frac{1}{\sqrt{2}}\bigl[\begin{smallmatrix} j & -j \\
1 & 1 \end{smallmatrix}\bigr]$. Let $\text{EIGV}( \Psi_{2\pi})$ be
the collection of M-MATRIX eigenvectors. It then follows from
inspection of \cref{app1:eq1} that

\begin{equation}
\text{EIGV}( \Psi_{2\pi}) = \bigl[ \hat{\phi} , \hat{r} , \{
\hat{w}_i\}_{i=1}^m , \{ \hat{z}_{2i} \pm j\hat{z}_{2i-1}
\}_{i=1}^{k}\bigr] \label{app1:eq2}
\end{equation}

where $\hat{\phi},\hat{r},\hat{w}_i$ \emph{etc.} represents the
coordinate vectors corresponding to coordinates $\bar{y} = ( \phi
, r , w, z )$ (see text in \cref{sec2b} and \cref{sec2b:foot2}).
The PNF-OSC Floquet decomposition $(u_1(s),u_2(s),\cdots,u_n(s))$,
at phase $\phi = s$, simply correspond to the eigenvectors for the
M-MATRIX (which are independent of $s$ due to symmetry).

\begin{equation}
(u_1(\tau),u_2(\tau),\cdots , u_n(\tau)) = ( \hat{\phi} , \hat{r}
, \{ \hat{w}_i\}_{i=1}^m , \{ \hat{z}_{2i} \pm j\hat{z}_{2i-1}
\}_{i=1}^{k}) \label{app1:eq3}
\end{equation}

where, following the notation introduced in \cref{sec1b} we let
$\tau \in [0,2\pi)$ index the phase of the Floquet vectors
(instead of $\phi$). The M-MATRIX eigenvalue-spectrum, also know
as the characteristic multiplies $\{\lambda_i\}_{i=1}^n$, follow
straight from inspection of \cref{app1:eq1}

\begin{equation}
\begin{gathered}
\textnormal{spec}\{\Psi_{2\pi}\} = \{\lambda_i\}_{i=1}^n =
\bigl[1 , \exp(-\mu 2\pi) , \{\exp( -2\pi\beta_i)\}_{i=1}^m, \\
\{\exp(-2\pi|\sigma_i| \pm j2\pi\nu_i)\}_{i=1}^k\bigr]
\end{gathered}
\label{app1:eq4}
\end{equation}

The PNF-OSC in \cref{sec2b:eq1} is hence able to generate any
possible M-MATRIX eigenvalue-spectrum (real modes + complex modes)
by simply varying the discrete parameters $m,k$ and the value of
circuit parameters $\{\mu , \beta_i , \sigma_i , \nu_i \}$.

\section*{Acronyms}
\begin{description}[leftmargin=8em,style=nextline,topsep=0pt,itemsep=-1ex,partopsep=1ex,parsep=1ex]
\item[LR] linear-response. \item[F-MATRIX] fundamental matrix map.
\item[M-MATRIX] Monodromy matrix (special F-MATRIX). \item[PSS]
periodic steady-state. \item[SYM-OFD] handle for the model
developed herein. \item[N-OSC] $n$-dimensional hyperbolic stable
oscillator. \item[NF-OSC] normal-form oscillator. \item[PNF-OSC]
polar normal-form oscillator. \item[OFD] orthogonal Floquet
decomposition.
\end{description}

\section*{Symbols}
\begin{description}[leftmargin=8em,style=nextline,topsep=0pt,itemsep=-1ex,partopsep=1ex,parsep=1ex]
\item[$\mathbf{o}$] the (P)NF-OSC, $\mathbf{o} =
(\psi_{\tau},\xi)$. \item[$\mathbf{q}$] the N-OSC, $\mathbf{q} =
(\phi_{\tau},\gamma)$. \item[$\psi_{\tau},\phi_{\tau}$]
(P)NF/N-OSC flow maps. \item[$\xi,\gamma$] (P)NF/N-OSC hyperbolic
limit-cycles. \item[$\mathbb{U},\mathbb{W}_s$] (P)NF/N-OSC stable
manifolds (domains). \item[$\mathbf{B}$] the set of all oscillator
tangent bundles. \item[$\mathbf{B}^{\perp}$] OFD tangent bundles
(subset of $\mathbf{B}$). \item[$\mathbf{O}$] the set of all
hyperbolic stable oscillators. \item[$\mathbf{O}^{\perp}$] set of
OFD oscillators, \emph{i.e.} circuits with OFD tangent bundles,
$\mathbf{B}^{\perp}$, (subset of $\mathbf{O}$).
\item[$\mathbf{H}$] set of smooth conjugation maps $h : \mathbb{U}
\to \mathbb{W}_s$. \item[$\mathbf{H}_C$] $h \in
\mathbf{H}_C\subset \mathbf{H}$ if $h\bigl|_{\xi}$ (restriction to
$\xi$) is conformal (angle preserving). \item[$\mathcal{SYM}_n$]
the set of all perfectly symmetric $n$-dimensional limit-cycles
(closed curves). \item[$\sim_h$] the conjugation operator
(equivalence relation) in-terms of the unique conjugation map $h$.
\item[$\mathbf{o}\sim_h \mathbf{q}$] conjugate oscillators
$\mathbf{o}$ and $\mathbf{q}$. \item[$\overset{c}{\sim}_h$] the
$\sim_h$ operator + restriction $h \in \mathbf{H}_C$.
\end{description}

\section*{Acknowledgment}

The authors gratefully acknowledge partial financial support by
German Research Foundation (DFG) (grant no. KR 1016/16-1).



\end{document}